\documentclass[journal]{IEEEtran}  
\IEEEoverridecommandlockouts
\usepackage{cite}
\usepackage{amsmath,amssymb,amsfonts}
\usepackage{algorithmic}
\usepackage{graphicx}
\usepackage{algorithm,algorithmic}
\usepackage{hyperref}
\hypersetup{hidelinks}
\usepackage{textcomp}

\newtheorem{problem}{\textbf{Problem}}
\newtheorem{proposition}{\bf{Proposition}}
\newtheorem{definition}{\bf{Definition}}
\newtheorem{lemma}{\bf{Lemma}}
\newtheorem{theorem}{\bf{Theorem}}
\newtheorem{remark}{\bf{Remark}}
\newtheorem{assumption}{\bf{Assumption}}

\newtheorem{corollary}{\bf{Corollary}}

\def\myunderbar#1{\underline{\sbox\tw@{$#1$}\dp\tw@\z@\box\tw@}}
\newenvironment{proof}
{\par\noindent\textit{Proof:}\ }
{\hfill$\blacksquare$\par}

\newcommand{\diag}{\text{diag}}

\newcommand{\vect}[1]{\boldsymbol{#1}}
\newcommand{\norm}[1]{ \lVert #1 \rVert}

\newcommand{\inN}[1]{\mathcal{N}_{#1}^{\text{in}}}
\newcommand{\outN}[1]{\mathcal{N}_{#1}^{\text{out}}}

\newcommand{\inNeightime}[1]{\mathcal{N}_{#1}^{\text{in}}(t)} 
\newcommand{\outNeightime}[1]{\mathcal{N}_{#1}^{\text{out}}(t)} 
\newcommand{\spec}[0]{\text{spec}}
\newcommand{\Outhop}[2]{\mathcal{O}_{#1}^{{#2}\text{-hop}}} 
\newcommand{\Inhop}[2]{\mathcal{I}_{#1}^{{#2}\text{-hop}}} 
\newcommand{\TNeigh}[1]{\mathcal{N}_{#1}^{\mathcal{T}}} 
\newcommand{\sat}[0]{\text{sat}}

\newcommand{\estimate}[3]{\hat{{#1}}^{{#2}}_{{#3}}}
\newcommand{\state}[3]{{{#1}}^{{#2}}_{{#3}}}
\newcommand{\error}[3]{\Tilde{{#1}}^{{#2}}_{{#3}}}
\newcommand{\statediff}[5]{{#1}^{#2}_{#3}-{#4}^{#5}_{#3}}


\usepackage{mathtools}

\def\BibTeX{{\rm B\kern-.05em{\sc i\kern-.025em b}\kern-.08em
    T\kern-.1667em\lower.7ex\hbox{E}\kern-.125emX}}
\usepackage{array}
\newcolumntype{P}[1]{>{\centering\arraybackslash}p{#1}}

\begin{document}
\title{Observer-Based Control for Multi-Agent Systems under Communication-Task Graph Mismatch and Directed Switching Topologies}
\author{Tommaso Zaccherini, Siyuan Liu, and Dimos V. Dimarogonas
\thanks{This work was supported in part by the Wallenberg AI, Autonomous Systems and
Software Program (WASP) funded by the Knut and Alice Wallenberg (KAW) Foundation, 
the ERC LEAFHOUND Project, and the Swedish Research Council (VR).
}
\thanks{Tommaso Zaccherini and Dimos V. Dimarogonas are with the Department
of Decision and Control Systems, KTH Royal Institute of Technology, Stockholm, Sweden.
	E-mail: {\tt\small \{tommasoz, dimos\}@kth.se}. Siyuan Liu is with the Department of Electrical Engineering, Control Systems Group, Eindhoven University of Technology, the Netherlands. Email: {\tt\small s.liu5@tue.nl}}
}

\maketitle
\begin{abstract}
This paper addresses decentralized estimation and control of heterogeneous multi-agent systems communicating over directed switching topologies. The objective is to satisfy cooperative tasks that require coordination among agents according to a time-invariant task graph. To address the potential mismatch between the communication and task graphs, each agent is equipped with a \mbox{$k$-hop} Prescribed Performance State Observer (\mbox{$k$-hop} PPSO) to estimate the states of agents up to $k$ communication hops away using only information exchanged with its \mbox{$1$-hop} neighbors. A novel prescribed performance funnel construction is developed to account for the switching communication topology, and a sufficient lower bound on the \mbox{dwell-time} is derived to guarantee the prescribed performance specifications under topology variations. Further analysis shows that, if the desired control objective is achieved with exact knowledge of the required nonlocal states and the resulting closed-loop system is uniformly set-ISS with respect to the state estimation error, then replacing these states with their estimates guarantees convergence to a neighborhood of the desired set. The size of this neighborhood can be made arbitrarily small by appropriately selecting the observer performance bounds. A simulation study demonstrates task satisfaction under directed switching topologies.
\end{abstract}


\section{Introduction}
\label{sec:introduction}
\IEEEPARstart{T}{he} increasing use of autonomous robotic systems in demanding and \mbox{large-scale} environments has motivated a shift from centralized to distributed architectures, where multiple agents coordinate using locally available information. Compared with single-agent systems, \mbox{multi-agent} systems (MAS) exploit distributed sensing, computation, and actuation to accomplish complex tasks cooperatively, with applications including formation maintenance \cite{XIAO_Finite_time_formation}, distributed sensing \cite{ConsensusFilter_Olfati}, and autonomous exploration \cite{MultiRobot_explo}.

Distributed MAS control typically assumes that the information required by each agent is available from its communication neighbors \cite{liu2018coordinate,yang2016decentralized,PanDecentralizedFormation,amirkhani2022consensus,oh2015survey}. However, task dependencies may extend beyond the communication graph, requiring agents to access information from non-neighbors. Existing approaches to this discrepancy can broadly be categorized into task decomposition and information reconstruction. Task-decomposition methods adapt tasks to the communication structure, but may require additional optimization or apply only to \mbox{distance-based} tasks \cite{WangTaskDecomp,MariaTaskGraph,GregTakDecom}. Alternatively, information-reconstruction methods exploit distributed sensing and state estimation to recover unavailable information from spatially distributed measurements \cite{Mitra_Distributed_Observer_for_LTI,Distributed_Kalman_filtering,yang2022state}. The resulting estimates support \mbox{observer-based} controllers for distributed coordination objectives \cite{ObsBasedConsensus,FU2016231,Yanf_PPC_ObsBase_Consensus,HONG2008846,7440802,Stamouli_Formation_Control_PPC_obs_for_formation,6308694,LI2011510,XU2013334,ZHAO201322}. However, many information-reconstruction methods are developed for linear or \mbox{single-double} integrator dynamics \cite{FU2016231,Mitra_Distributed_Observer_for_LTI,XU2013334,LI2011510}, and nonlinear \mbox{observer-based} control approaches are often tailored to specific objectives such as consensus \cite{Yanf_PPC_ObsBase_Consensus,ZHAO201322}, leader-following \cite{XU2013334}, or formation control \cite{Stamouli_Formation_Control_PPC_obs_for_formation}. These limitations leave a gap in addressing general task-dependent coordination of heterogeneous nonlinear MAS, where the behaviors of multiple agents may be coupled through nonlocal information dependencies.

To accommodate nonlocal information dependencies without task decomposition, our recent work \cite{ZaccheriniLCSS,ZaccheriniRobustkhop, zaccherini2026controlmultiagentsystemsstl} introduced distributed and decentralized \mbox{$k$-hop} Prescribed Performance Observers ($k$-hop DPPOs) for nonlinear MAS. These observers enable each agent to estimate the states of \mbox{non-neighboring} agents that are two or more hops away using only $1$-hop communication, while guaranteeing prescribed bounds on the resulting state estimation errors. The above results, however, rely on a \mbox{time-invariant}, undirected communication topology, under which the communication paths used to propagate the \mbox{$k$-hop} estimates remain fixed and bidirectional. This assumption limits the applicability of the \mbox{$k$-hop} DPPOs to communication networks in which links are fixed and information exchange is symmetric. In more general MAS, communication links may instead be asymmetric and vary over time, resulting in directed and time-varying topologies.

To reconstruct nonlocal states despite topology switching and bounded disturbances, this paper develops a decentralized \mbox{$k$-hop} prescribed performance state observer (\mbox{$k$-hop} PPSO) for heterogeneous MAS over directed and switching communication graphs. As in the undirected case, the proposed observer enables each agent to estimate the states of agents located two to $k$ hops away, while keeping the estimation errors within predefined performance bounds specified at the design stage. Directed communication graphs, however, generally have nonsymmetric Laplacian matrices, which fundamentally change the structural properties of the error dynamics used to establish prescribed-performance guarantees. Moreover, topology switching changes the communication paths along which information propagates, resulting in time-varying \mbox{$k$-hop} neighborhoods. As a result, agents may join or leave a given agent's \mbox{$k$-hop} neighborhood as the communication topology switches, causing the set of estimated agents to change over time. Since the estimation vector is defined over this set of agents, its dimension changes accordingly, resembling the behavior of Open Multi-Agent Systems (OMAS). Thus, in order to define the vector of estimates consistently as agents enter and leave the estimated set, we draw inspiration from the OMAS literature \cite{Xue_Tang_Ren_Qian_2022}.

Beyond the observer design, we establish the applicability of the proposed observer to general control objectives through a \mbox{set-ISS} characterization of the \mbox{closed-loop} system. Specifically, we show that if the \mbox{closed-loop} dynamics are uniformly \mbox{set-ISS} with respect to the state estimation error, replacing nonlocal states with their estimates preserves convergence to a neighborhood of the desired set, whose size can be made arbitrarily small through an appropriate choice of the observer performance bounds.

To the best of our knowledge, this work presents the first decentralized  \mbox{$k$-hop} Prescribed Performance State Observer and associated \mbox{observer-based} control framework for heterogeneous MAS over directed switching communication topologies, with applicability to general control objectives. The contribution of the paper is summarized as follows:
\begin{itemize}
    \item Compared with the observer-based approaches in \cite{Yanf_PPC_ObsBase_Consensus,ZHAO201322,XU2013334}, which target specific coordination objectives, the proposed observer is independent of the particular task. It also explicitly prescribes the transient performance of the state estimation error, yielding a priori bound on the worst-case state estimation error at every time instant. This generality is particularly relevant for applications in which task satisfaction depends on guaranteed bounds on the estimation uncertainty throughout the transient, as in \cite{zaccherini2026controlmultiagentsystemsstl}.

    \item Unlike \cite{Stamouli_Formation_Control_PPC_obs_for_formation}, which considers formation control based on a prescribed-performance observer for single-integrator MAS communicating over an undirected time-invariant graph, the proposed \mbox{$k$-hop} PPSO is developed for heterogeneous MAS under directed switching communication topologies. Moreover, rather than reconstructing aggregate formation information through a dynamic average consensus scheme, the proposed \mbox{$k$-hop} PPSO reconstructs the states of nonlocal agents required by the task. The choice of $k$ determines the maximum \mbox{communication-hop} distance over which states can be reconstructed, allowing the observer to support task dependencies beyond each agent's immediate communication neighborhood.

    \item A novel prescribed-performance funnel construction is developed to account for topology switching in the \mbox{$k$-hop} state observer. Unlike existing PPC approaches for switched systems \cite{Rovithakis_PPC_Directed_switching_Graphs,PPC_SwitchingTopologySingularityFree}, which either relax the prescribed funnels online at switching instants or construct an over-approximating funnel a posteriori, the proposed method constructs auxiliary funnels online based on the worst-case variation of the \mbox{$k$-hop} neighborhood. Under suitable feasibility and minimum dwell-time conditions, these funnels ensure that the original prescribed performance specifications remain satisfied across topology switches.
\end{itemize}

The remainder of the paper is organized as follows. Section~\ref{Sec: Preliminaries} introduces the notation, preliminaries, and problem formulation. Section~\ref{Disagreement dynamics} defines the disagreement vectors among the agents' estimates and derives their dynamics.
Section~\ref{Sec: Decentralized PPSO} presents the proposed \mbox{$k$-hop} PPSO and the conditions under which the original prescribed performances are satisfied under switching topologies.
Section~\ref{Sec: Closed-loop} presents the feedback control structure and establishes the conditions under which the \mbox{$k$-hop} \mbox{estimation-based} feedback controller guarantees convergence to the team objective.
Section~\ref{Sec: Case Study} demonstrates the effectiveness of the proposed approach through simulation results, and Section~\ref{Sec: Conclusions} concludes the paper with final remarks and directions for future work.

\section{Preliminaries and Problem Setting}\label{Sec: Preliminaries}
\textbf{Notation:} 
Let $\mathbb{R}$, $\mathbb{R}_{\ge 0}$, and $\mathbb{R}_{>0}$ denote the sets of real, nonnegative real, and strictly positive real numbers, respectively. Furthermore, let $\mathbb{N}$ denote the set of natural numbers including $0$, and $\mathbb{C}$ the set of complex numbers. The Euclidean space of dimension $n$ is $\mathbb{R}^n$, and $\mathbb{R}^{n\times m}$ represents the space of real matrices with $n$ rows and $m$ columns. The identity and the null matrices of dimension $n$ are $I_n$ and $0_n$, and $1_n$ denotes the $n$-dimensional vector of ones. 
For any set $S$, $|S|$ denotes its cardinality, $S^c$ its complement, and $\partial S$ its boundary. For $N$ sets $\{S_1,\dots, S_N\}$, the Cartesian product, intersection, and union are $\bigtimes^N _{i=1} S_i$, $\bigcap^N _{i=1} S_i$, and $\bigcup^N _{i=1} S_i$, respectively. For a set $S = \{s_1,\ldots,s_N\}$ with $s_i \in \mathbb{R}$, $S_{\max} := \max_{s_i\in S} s_i$ and $S_{\min} := \min_{s_i\in S} s_i$. 
Given a set of vectors $\{s_1, \dots, s_N\}$, with $s_i \in \mathbb{R}^n$, $[s_i]_{i \in \{1,\dots, N\}}:= [s_1^\top, \dots, s_N^\top]^\top$.
For $B \in \mathbb{R}^{n\times n}$, define with $\spec(B)$ the spectrum of $B$, and let $\sigma_{\min}(B)$ and $\sigma_{\max}(B)$ denote the smallest and largest singular values of $B$. We use $\|B\|$ and $\|B\|_\infty$ to denote its induced $2$-norm and $\infty$-norm, respectively. For $B \in \mathbb{R}^{n\times m}$, $[B]_{i,j}$ denotes the \mbox{$(i,j)$-th} entry of $B$. For $x \in \mathbb{R}^n$, the Euclidean norm is $\|x\| = \sqrt{x^\top x}$, and the $j$-th component is $[x]_j$. The diagonal matrix with diagonal entries $a_1,\ldots,a_n$ is $\mathrm{diag}(a_j)_{j = 1}^n$. Similarly, given a set of indexes $\mathcal{I} = \{1, \dots, n\}$, $\mathrm{diag}(a_i)_{i \in \mathcal{I}} := \mathrm{diag}(a_j)_{j = 1}^n$. 
$\otimes$ denotes the Kronecker product. A function $f$ is of class $\mathcal{C}^1$ if it is continuously differentiable on its domain. We define functions $\mathcal{K}$ and $\mathcal{KL}$ as follows: $\mathcal{K} \!=\! \{\gamma : \mathbb R_{\ge 0}\rightarrow\mathbb R_{\ge 0}   :  \gamma \text{ is continuous, strictly increasing and } \gamma(0)=0\}$;  $\mathcal{KL} \!=\! \{\beta : \mathbb R_{\ge 0} \!\times \mathbb R_{\ge 0} \rightarrow\mathbb R_{\ge 0}  :$ for each fixed $s$, $\beta(r,s) \in \mathcal{K}$  with  respect to  $r$ and, for each fixed $r > 0$, $\beta(r,s)$ is decreasing with respect to  $s$  and $\lim_{ s \rightarrow \infty } \beta(r,s) = 0\}$.

\subsection{Multi-agent systems}
Consider a heterogeneous multi-agent system (MAS) consisting of $N$ agents $\mathcal{V} :=\{1, \dots, N\}$.
Let $n = \sum_{i=1}^{N} n_i$ and $m = \sum_{i=1}^{N} m_i$, and define the global state and input vectors as $\vect{x}:= [x_i]_{i \in \mathcal{V}} \in \mathbb{R}^n $ and $\vect{u}:=[u_i]_{i \in \mathcal{V}} \in \mathbb{R}^m$. The dynamics of each agent $i \in \mathcal{V}$ is given by:
\begin{equation}
    \label{eq: agent's dynamic}
    \dot x_i(t) =  f_i(x_i(t)) + g_i(u_i(t)) + w_i(\vect{x},t),
\end{equation}
where $x_i \in \mathbb{R}^{n_i}$ and $u_i \in \mathbb{R}^{m_i}$ denote the state and input of agent $i$, respectively, $f_i:\mathbb{R}^{n_i} \rightarrow \mathbb{R}^{n_i}$ describes the intrinsic dynamics, $g_i:\mathbb{R}^{m_i} \rightarrow \mathbb{R}^{n_i}$ is the input map, and $w_i: \bigtimes^N _{i=1}\mathbb{R}^{n_i} \times \mathbb{R}_{\geq 0} \rightarrow \mathbb{R}^{n_i}$ models coupling effects and exogenous disturbances acting on agent $i$.

\begin{assumption}
    \label{Assumption on existence of a solution}
    For each $i\in\mathcal V$: (i) $f_i:\mathbb{R}^{n_i} \rightarrow \mathbb{R}^{n_i}$ is locally Lipschitz; (ii) $g_i:\mathbb{R}^{m_i} \rightarrow \mathbb{R}^{n_i}$ is measurable and essentially locally bounded; (iii) $w_i: \bigtimes^N _{i=1}\mathbb{R}^{n_i} \times \mathbb{R}_{\geq 0} \rightarrow \mathbb{R}^{n_i}$ is continuous and uniformly bounded on $\bigtimes^N _{i=1}\mathbb{R}^{n_i} \times \mathbb{R}_{\geq 0}$.
\end{assumption}

\begin{assumption}
    \label{Assumption on g}
    (i) For each $i$ there exists a continuous function $\ell_i: \mathbb{R}_{\geq 0} \rightarrow \mathbb{R}_{\geq 0}$ such that $\norm{f_i(x) - f_i(y)} \leq \ell_i(\norm{x-y})$ for all $x,y \in \mathbb{R}^{n_i}$. (ii) $g_i(u_i(t))$ is uniformly bounded for all $t\geq 0$.
\end{assumption}
\begin{remark}
    \label{remark: agent trajectories exist}
    Splitting the segment from $0$ to $x$ into $\lceil \norm{x} \rceil$ pieces of length at most one, where $\lceil r \rceil := \min\{k \in \mathbb{N} : k \geq r\}$, gives, under Assumption~\ref{Assumption on g}(i), the linear growth bound $\norm{f_i(x)} \leq a_i + C_i \norm{x}$, with $C_i := \max_{r \in [0,1]} \ell_i(r)$ and $a_i := \norm{f_i(0)} + C_i$. Since $g_i(u_i(t))$ and $w_i$ are bounded, the Gr\"onwall-Bellman inequality \cite[Lemma~A.1]{khalil2002nonlinear} precludes finite escape. Hence, under Assumptions~\ref{Assumption on existence of a solution}-\ref{Assumption on g}, the agent trajectories exist for all $t \geq 0$.
\end{remark}

\subsection{Communication graph} 
\label{Section: communication graph}
Let $\mathbb{G} := \{\mathcal{G}_1, \dots, \mathcal{G}_M\}$ denote the set of admissible directed communication topologies, where each graph $\mathcal{G}_l := (\mathcal{V}, \mathcal{E}_l)$, $l \in \{1,\dots, M\}$, represents the information exchange structure among the $N$ agents. Specifically, an edge $(i,j) \in \mathcal{E}_l$ indicates that agent $j$ receives information from agent $i$ under the communication topology $\mathcal{G}_l$.

\begin{remark}
    Since each unordered pair of distinct agents admits one of four mutually exclusive edge configurations (no edge, either directed edge, or both directed edges), the number of distinct communication topologies on $N$ agents is upper bounded by $4^{N(N-1)/2}$. Therefore, $\mathbb{G}$ is finite.
\end{remark}

The active communication topology at time $t$ is denoted by $\mathcal{G}_{\sigma(t)}$, where $\sigma: [0, \infty) \rightarrow  \{1, \dots, M\}$ is a right-continuous piecewise-constant switching signal describing the evolution of the communication topology. Let $t_{\tau}$, $\tau \in \mathbb{N}$, denote the switching instants.

For the active graph $\mathcal{G}_{\sigma(t)}$, define the \mbox{in-} and \mbox{out-neighbor} sets of agent $i$ as $\inNeightime{i} := \{j\in\mathcal{V}: (j,i) \in \mathcal{E}_{\sigma(t)}\}$ and $\outNeightime{i} := \{j\in\mathcal{V}: (i,j) \in \mathcal{E}_{\sigma(t)}\}$, respectively. The corresponding \mbox{in-} and \mbox{out-degrees} at time $t$ are $d^{\text{in}}_{i}(t) := |\inNeightime{i}|$ and $d^{\text{out}}_{i}(t) := |\outNeightime{i}|$. 

A directed path from $i$ to $j$ is defined as a sequence of directed edges connecting $i$ to $j$ without revisiting any vertex. For $i,j \in \mathcal{V}$ and $l \in \{1,\dots,M\}$, let $D^{ij}_l$ denote the length of the
shortest directed path from $i$ to $j$ in $\mathcal{G}_l$, and $D^{ij}_{\sigma(t)}$ the corresponding length in the active graph. For $l \in \{1,\dots,M\}$, the $k$-hop out-neighbor set of agent $i$ in $\mathcal{G}_l$ is $\mathcal{O}_i^{k\text{-hop}}(\mathcal{G}_l) := \{j \in \mathcal{V}\setminus\{i\} : 2 \leq D^{ij}_l \leq k\}$, and the $k$-hop out-neighbor set at time $t$ is $\mathcal{O}_i^{k\text{-hop}}(t) := \mathcal{O}_i^{k\text{-hop}}(\mathcal{G}_\sigma(t))$, with cardinality $\eta_i(t) = |\mathcal{O}_i^{k\text{-hop}}(t)|$. Similarly, the corresponding $k$-hop in-neighbor set of $i$ is $\mathcal{I}_i^{k\text{-hop}}(t) := \{j \in \mathcal{V} : i \in \mathcal{O}_j^{k\text{-hop}}(t)\}$. Let $\mathcal{O}_i^{\cup} := \bigcup_{l=1}^{M} \mathcal{O}_i^{k\text{-hop}}(\mathcal{G}_l)$ if $\mathbb{G}$ is known, and $\mathcal{O}_i^{\cup} := \mathcal{V}\setminus\{i\}$ otherwise, denote the set of agents that can belong to the $k$-hop out-neighborhood of agent $i$ under some admissible topology. By construction, $\mathcal{O}_i^{k\text{-hop}}(t) \subseteq \mathcal{O}_i^{\cup}$ for all $t \geq 0$.

The \mbox{adjacency} matrix of $\mathcal{G}_{\sigma(t)}$ is defined as 
\begin{equation*}
    [\mathcal{A}^{\text{in}}(t)]_{i,j} = \begin{cases}
        1, \quad (j,i) \in \mathcal{E}_{\sigma(t)}, \\
        0, \quad \text{otherwise}.
    \end{cases}
\end{equation*}
Furthermore, $\mathcal{D}^{\text{in}}(t) := \diag(d^{\text{in}}_{i}(t))_{i=1}^N$, and the in-degree Laplacian of $\mathcal{G}_{\sigma(t)}$ is $L^{\text{in}}(t) := \mathcal{D}^{\text{in}}(t)- \mathcal{A}^{\text{in}}(t)$. 

\begin{assumption}
    \label{Assumption on neighborhood}
    Each agent $i\in \mathcal{V}$ has knowledge of $\inNeightime{i}, \outNeightime{i}$ and $\Outhop{i}{k}(t)$.
\end{assumption}
\begin{assumption}
    \label{Assumption on 2-hop communication}
    At each time instant, each agent $i \in \mathcal{V}$ has access to the states and inputs of its in-neighbors and can relay this information to its out-neighbors.
\end{assumption}

Assumption~\ref{Assumption on neighborhood} can be satisfied in practice, as distributed neighborhood discovery algorithms have been extensively studied in the sensor network literature \cite{994183}.
Furthermore, Assumption~\ref{Assumption on 2-hop communication} holds, for example, when each agent acquires the states of its $1$-hop in-neighbors through onboard sensing and transmits this information to its out-neighbors. As discussed in Remark~\ref{Remark on local disagreement vector computation and positive definiteness of Mik}, Assumption~\ref{Assumption on 2-hop communication} can be relaxed.

To simplify notation, the dependence of \mbox{graph-related} quantities on $t$ through the switching signal $\sigma(t)$ is omitted whenever it is clear from the context. In the following, at each time $t$, we consider only agents with a nonempty \mbox{$k$-hop} \mbox{out-neighbor} set. 

\subsection{State estimates and their transitions}
\label{Section: state and input estimate}
Let $\vect{x}^i$ denote the stacked vector of states associated with the \mbox{$k$-hop} \mbox{in-neighbors} of agent $i$, i.e., $\state{\vect{x}}{i}{} := \left[\state{x}{}{j}\right]_{j \in \Inhop{i}{k}}$.
Furthermore, let $\estimate{\vect{x}}{i}{} := \left[\estimate{x}{i}{j}\right]_{j \in \Inhop{i}{k}}$ denote the corresponding estimate, where $\estimate{x}{i}{j}$ denotes the estimate of $x_j$ maintained at agent $i$. The associated estimation error is defined as:
\begin{equation}
    \label{Definition: error on input and state estimation perfromed by i}
    \error{\vect{x}}{i}{} :=\left[\error{x}{i}{j}\right]_{j \in \Inhop{i}{k}},
\end{equation}
where $\error{x}{i}{j} := \estimate{x}{i}{j} -\state{x}{}{j}$ for all $j \in \Inhop{i}{k}$. 

Define $\vect{x}_i := 1_{\eta_i} \otimes x_i$. Moreover, let
\begin{equation}
    \label{eq: estimation of state and input of agent i}
    \estimate{\vect{x}}{}{i} := \left[\estimate{x}{j}{i}\right]_{j \in \Outhop{i}{k}}
\end{equation}
denote  the stacked vector of estimates of $x_i$ maintained by its \mbox{$k$-hop} \mbox{out-neighbors} $\Outhop{i}{k}$. Similarly to \eqref{Definition: error on input and state estimation perfromed by i}, define with $\error{\vect{x}}{}{i} := \estimate{\vect{x}}{}{i} -\vect{x}_i$ the estimation error maintained by each $j \in \Outhop{i}{k}$, i.e.,
\begin{equation}
    \label{Definition: error on input and state estimation of agent i}
        \error{\vect{x}}{}{i} := \left[\error{x}{j}{i}\right]_{j \in \Outhop{i}{k}}, 
\end{equation}
where $\error{x}{j}{i} := \estimate{x}{j}{i} - x_i$.

Since the active communication topology switches from $\mathcal{G}_{\sigma(t_{\tau-1})}$ to $\mathcal{G}_{\sigma(t_{\tau})}$ at the switching instant $t_{\tau}$, the equality $\Outhop{i}{k}(t_{\tau-1}) = \Outhop{i}{k}(t_{\tau})$ does not necessarily hold. Consequently, the dimension of the state estimation error vector $\error{\vect{x}}{}{i}$ may change at switching instants. This differs from the standard switched system framework \cite{Liberzon_2003}, where the state dimension is fixed and state trajectories evolve in a fixed dimensional state space. Therefore, an explicit mapping between the pre- and post-switch estimation error vectors is required. 
To account for the varying state dimension, we adopt notation similar to \cite{Xue_Tang_Ren_Qian_2022}. Specifically, the evolution of the state estimation error $\error{\vect{x}}{}{i}$ across a switching instant is modeled as
\begin{equation}
    \label{eq: State estimation error after the switch}
    \error{\vect{x}}{}{i}(t_{\tau}^+) = \Xi_i^\tau \ \error{\vect{x}}{}{i}(t_{\tau}^-) + \phi_i(t_{\tau}),
\end{equation}
where $\error{\vect{x}}{}{i}(t_{\tau}^-) \in \mathbb{R}^{n_i\eta_i(t_{\tau}^-)}$ and $\error{\vect{x}}{}{i}(t_{\tau}^+)\in \mathbb{R}^{n_i\eta_i(t_{\tau}^+)}$ denote the estimation error vectors before and after the switch at $t_{\tau}$, respectively, $\Xi_i^\tau \in \mathbb{R}^{n_i\eta_i(t_{\tau}^+)\times n_i\eta_i(t_{\tau}^-)}$ is a selection matrix that extracts and reorders the components corresponding to agents remaining in the $k$-hop neighborhood, and $\phi_i(t_{\tau})$ collects the initialization errors associated with agents entering the $k$-hop neighborhood after the switch.

For notational simplicity, in the following analysis we consider $n_i=1$ for all $i\in \mathcal{V}$. The extension to higher-dimensional states follows by replacing the involved matrices with their Kronecker-product extensions.

\subsection{Problem formulation}
\label{Problem formulation}
For every agent $i \in \mathcal{V}$ and for all $j \in \mathcal{O}_i^{\cup}$, let $\delta_i^j: \mathbb{R}_{\geq 0} \rightarrow \mathbb{R}_{>0}$ be a prescribed performance function used to capture the performance bounds for the estimation errors.
\begin{definition}
    \label{Definition of Prescribed performance function}
    $\rho :\mathbb{R}_{\geq 0} \rightarrow \mathbb{R}_{> 0}$ is a prescribed performance function if for all $t \in \mathbb{R}_{\geq 0}$: (i) $\rho \in \mathcal{C}^1$; (ii) $\rho_L\leq \rho(t) \leq \rho_U$ for some $ \rho_U < \infty $ and $\rho_L > 0$; (iii) $|\dot{\rho}(t)|\leq \rho^d_U$ for some $\rho^d_U<\infty$. 
\end{definition}

One conventional choice is $\rho(t)= (\rho(0)-\rho(\infty)) e^{-l t} + \rho(\infty)$ where $\rho(0)$ and $\rho(\infty)$ denote the initial and steady-state values with $\rho(0) \geq \rho(\infty) > 0$, and $l \geq 0$ specifies the decay rate.

Then, the goal of this work is formulated as follows.
\begin{problem}
    \label{Problem state estimation error convergence}
     Given the heterogeneous MAS \eqref{eq: agent's dynamic} communicating over the directed switching graphs $\mathcal{G}_{\sigma(t)} \in \mathbb{G}$ and prescribed performance functions $\delta^{j}_i(t)$, design a decentralized \mbox{$k$-hop} observer that guarantees $|\error{x}{j}{i}(t)| <\delta^{j}_i(t)$ for all $t \geq 0$, $i \in \mathcal{V}$, and $j \in \Outhop{i}{k}(t)$. 
     Establish sufficient conditions on the observer initialization and topology switching for these bounds to hold. Finally, when the resulting estimates are used for control, characterize how the estimation errors affect the distance to the desired target set.
\end{problem}

\section{Disagreement dynamics and problem reformulation}\label{Disagreement dynamics}
For each $i \in \mathcal{V}$ and $j \in \Outhop{i}{k}$, define the disagreement term $\xi_i^{j}$ on the estimate of $x_i$ performed by $j$ as:
\begin{equation}
    \label{eq: disagreement vector on the estimate of agent i}
    \xi_i^{j}(t) :=\sum_{l\in(\inN{j}{} \cap \Outhop{i}{k})} (\statediff{\hat{x}}{j}{i}{\hat{x}}{l}) + |\inN{j}{} \cap \outN{i}{}|  (\hat{x}^{j}_i - x_i),
\end{equation}
where $\xi_i^{j}$ captures how $\hat{x}^{j}_i$ differs from: (i) the state estimate $\estimate{x}{l}{i}$ shared by those agents $l \in \inNeightime{j}{} \cap \Outhop{i}{k}(t)$  and (ii) the true state information $x_i$ shared by the agents $l \in \inNeightime{j}{} \cap \outNeightime{i}{}$.

\subsection{Disagreement vector}
\label{Section: disagreement vector}
By stacking the vectors $\xi_i^{j}$ for all $j \in \Outhop{i}{k}$, and using the estimation error vector defined in \eqref{Definition: error on input and state estimation of agent i}, the aggregate disagreement vector $\vect{\xi}_i := [\xi_{i}^{j}]_{j \in \Outhop{i}{k}}$ can be written as:
\begin{equation}
    \label{eq: disagreement vector expression}
    \vect{\xi}_i = (L_i(t) + H_i(t)) \error{\vect{x}}{}{i} = M_i(t) \error{\vect{x}}{}{i},
\end{equation}
where $L_i(t)= L_i(\mathcal{G}^i_{\sigma(t)})$ denotes the in-degree Laplacian of the subgraph $\mathcal{G}^i_{\sigma(t)} := ( \Outhop{i}{k}(t), \mathcal{E}^i_{\sigma(t)})$ induced by the \mbox{$k$-hop} out-neighbors of agent $i$, with $\mathcal{E}^i_{\sigma(t)} := \{(j,q) \in \mathcal{E}_{\sigma(t)}: j,q \in \Outhop{i}{k}(t)\}$, $H_i(t)= H_i(\mathcal{G}^i_{\sigma(t)}) := \text{diag}(|\inNeightime{j}{} \cap \outNeightime{i}{}|)_{j \in \Outhop{i}{k}}$ is the associated pinning matrix, and $M_i(t) = M_i(\mathcal{G}^i_{\sigma(t)}) := L_i(t) + H_i(t)$. 

The subgraph $\mathcal{G}^i_{\sigma(t)}$ can be decomposed into $l^i(t)$ weakly connected components $\{\mathcal{G}_{\sigma(t)}^{i,(1)}, \dots, \mathcal{G}_{\sigma(t)}^{i,(l^i)}\}$ \cite[Sec. 6.11]{newman2018networks}, where $\mathcal{G}_{\sigma(t)}^{i,(p)} := (\mathcal{V}_{\sigma(t)}^{i,(p)}, \mathcal{E}_{\sigma(t)}^{i,(p)})$, with $\mathcal{V}_{\sigma(t)}^{i,(p)} \subseteq \Outhop{i}{k}(t)$ and $\mathcal{E}_{\sigma(t)}^{i,(p)} := \{(j,q) \in \mathcal{E}^i_{\sigma(t)}: j,q \in \mathcal{V}_{\sigma(t)}^{i,(p)}\}$. Since distinct weakly connected components contain no edges between them, the off-diagonal entries of $L_i$ associated with pairs of nodes in different components are zero, while all in-neighbors of each node belong to the same component. Hence, there exists a permutation matrix $P^i_{\sigma(t)}$ that orders the nodes according to their weakly connected components such that $P^i_{\sigma(t)} L_i P^{i \top}_{\sigma(t)} = \diag(L_i^{(p)})_{p = 1}^{l^i}$ and $P^i_{\sigma(t)} H_i P^{i \top}_{\sigma(t)}= \diag(H_i^{(p)})_{p = 1}^{l^i}$, where $L_i^{(p)}$ and $H_i^{(p)}$ denote the in-degree Laplacian and pinning matrix associated with $\mathcal{G}_{\sigma(t)}^{i,(p)}$, respectively. Consequently, $P^i_{\sigma(t)} M_i P^{i \top}_{\sigma(t)} = \diag(M_i^{(p)})_{p = 1}^{l^i}$ with $M_i^{(p)} = L_i^{(p)} + H_i^{(p)}$, and
\begin{equation}
    \label{eq: Spectrum of M_i equal to the one of the block diagonals}
    \spec(M_i) = \bigcup_{p = 1}^{l^i} \spec(M_i^{(p)}).
\end{equation}
Hence, the spectral analysis of $M_i$ reduces to that of the block matrices $M_i^{(p)}$.

The following lemma establishes graph-theoretic properties of the \mbox{$k$-hop} \mbox{out-neighborhood} of agent $i$ that will be used to characterize $\spec(M_i)$.
\begin{lemma} 
    \label{lemma: neighbors}
    For all $i \in \mathcal{V}$ and all $\mathcal{G}_{\sigma(t)} \in \mathbb{G}$: (i) $\inN{j}{} \cap \Outhop{i}{k} \neq \emptyset\  \text{ or } \ \inN{j}{}\cap \outN{i}{}\neq \emptyset$ for all $j \in \Outhop{i}{k}$;
    (ii) every weakly connected component $\mathcal{G}_{\sigma(t)}^{i,(p)}$ of $\mathcal{G}^i_{\sigma(t)}$ contains at least one agent $j \in \Outhop{i}{k}$ satisfying $\inN{j}{}\cap \outN{i}{}\neq \emptyset$.
\end{lemma}
\begin{proof}
    Given in Appendix~\ref{Appendix: proof lemma neihbors}.
\end{proof}

The above graph property guarantees that every block matrix $M_i^{(p)}$ contains at least one pinned node, i.e., an agent $j \in \mathcal{V}_{\sigma(t)}^{i,(p)}$ with $[H_i]_{jj} > 0$, so that $H_i^{(p)} \neq 0$. Combining this property with Gershgorin's theorem \cite[Thm. 6.1.1]{horn2012matrix} yields the following spectral characterization of $M_i$.
\begin{lemma}
    \label{Lemma: -M_i is Hirwitz}
   $\Re(\lambda) > 0$ for all $\lambda \in \spec(M_i)$, i.e.,  $- M_i$ is Hurwitz.
\end{lemma}
\begin{proof}
    Given in Appendix~\ref{Appendix: proof lemma -M_i is Hirwitz}.
\end{proof}
\begin{remark}
    \label{Remark on local disagreement vector computation and positive definiteness of Mik}
    If Assumption~\ref{Assumption on 2-hop communication} fails, \eqref{eq: disagreement vector on the estimate of agent i} is not locally computable for those $j$ with $|\inN{j}{} \cap \outN{i}{}| \neq 0$, as they lack access to $x_i$ without \mbox{$2$-hop} information flow. Yet, local computability can be restored, while preserving the spectral properties on $M_i$, by augmenting the \mbox{$k$-hop} neighborhood with $1$-hop neighbors. Then, $L_{i}(t)$ is the Laplacian matrix of the subgraph induced by the new \mbox{$k$-hop} \mbox{out-neighborhood} of agent~$i$ at time $t$, and $H_i(t):= \text{diag}(h^{j}_i(t))_{j \in \Outhop{i}{k}(t)}$, where $h^{j}_i(t) = 1$ if $j \in \outN{i}{}(t)$, and $h^{j}_i(t) = 0$ otherwise.
\end{remark}

By Lemma~\ref{Lemma: -M_i is Hirwitz}, all eigenvalues of $M_i$ have positive real parts. Furthermore, $M_i$ has nonpositive off-diagonal entries. Therefore, by Theorem~\ref{theorem: M-matrix}, $M_i$ is a nonsingular $M$-matrix.
\begin{theorem}\noindent\cite[Thm.~2.5.2]{michel1977qualitative}
    \label{theorem: M-matrix}
    Let $D \in \mathbb{R}^{n\times n}$ be a matrix with $[D]_{ij} \leq 0$ for all $i \neq j$. Then the following statements are equivalent: i) $D$ is an $M$-matrix; ii) $D$ is nonsingular, and all elements of $D^{-1}$ are nonnegative; iii) the real parts of all eigenvalues of $D$ are positive.
\end{theorem}

The following property of nonsingular $M$-matrices will be used to construct a diagonal Lyapunov weighting matrix.
\begin{corollary}\noindent\cite[Corollary.~2.5.6]{michel1977qualitative}
    \label{Corollary: existence of a diagonal matrix P that render an M-matrix Hurwitz}
    Let $D \in \mathbb{R}^{n\times n}$ be a matrix with nonpositive \mbox{off-diagonal} elements. Then $D$ is an $M$-matrix if and only if there exists a diagonal matrix $A = \diag(\alpha_i)_{i=1}^n$, $\alpha_i > 0$, $i \in \{ 1, \dots, n\}$, such that the matrix $B = AD + D^\top A$ is positive definite.
\end{corollary}

\subsection{State-to-disagreement performance bounds}\label{Section: State-to-disagreement performance bounds}
To guarantee prescribed performance on the state estimation errors, the observer
must be designed so that
\begin{equation}
    \label{eq: state estimation error inequality}
    |\error{x}{j}{i}(t)| <\delta^{j}_i(t)
\end{equation}
holds for all $i \in \mathcal{V}$ and $j \in \Outhop{i}{k}$, where $\delta^{j}_i$ is a prescribed performance function as in Definition~\ref{Definition of Prescribed performance function}.

Since $M_i$ is nonsingular by Lemma~\ref{Lemma: -M_i is Hirwitz}, \eqref{eq: disagreement vector expression} gives $\error{\vect{x}}{}{i} = M_i^{-1} \vect{\xi}_i$. Therefore, $\norm{\error{\vect{x}}{}{i}} \leq \norm{M_i^{-1}} \norm{\vect{\xi}_i} = \sigma^{-1}_{\min}(M_i) \norm{\vect{\xi}_i}$, where the equality $\norm{M_i^{-1}} = \sigma^{-1}_{\min}(M_i)$ follows under the induced Euclidean norm. 

Since $|\error{x}{j}{i}(t)| \leq \norm{\error{\vect{x}}{}{i}(t)}$ for all $j \in \Outhop{i}{k}$, a sufficient condition for \eqref{eq: state estimation error inequality} is $\norm{\vect{\xi}_i(t)} < \sigma_{\min}(M_i) \min_{j \in \Outhop{i}{k}} \{\delta_i^{j}(t)\}$. Accordingly, the observer is designed to enforce
\begin{equation}
    \label{eq: definition of the bounds on the disagreement term}
    |\xi_i^{j}(t)| < \rho_i^{j}(t),
\end{equation}
for all $j \in \Outhop{i}{k}$, where the prescribed performance functions $\rho_i^{j}(t)$ are chosen such that $\vect{\rho}_i(t) = [\rho^{j}_i(t)]_{j  \in \Outhop{i}{k}}$ satisfies
\begin{equation}
    \label{eq: constraint on the prescribed prformance functions}
    \norm{\vect{\rho}_i(t)} \leq \sigma_{\min}(M_i) \min_{j \in \Outhop{i}{k}} \{\delta_i^{j}(t)\}
\end{equation}
with $\rho^{j}_i(0) \geq \rho^{j}_i(\infty) > 0$ and $\rho^{j}_i(0) > |\xi_{i}^{j}(0)|$. Since \eqref{eq: definition of the bounds on the disagreement term} implies
$\norm{\vect{\xi}_i(t)}<\norm{\vect{\rho}_i(t)}$, and $\norm{\error{\vect{x}}{}{i}} \leq \sigma^{-1}_{\min}(M_i) \norm{\vect{\xi}_i}$, the design condition \eqref{eq: constraint on the prescribed prformance functions} guarantees \eqref{eq: state estimation error inequality}.

The following lemma is useful for constructing prescribed performance functions satisfying \eqref{eq: constraint on the prescribed prformance functions}.
\begin{lemma}
    \label{Lemma: sum of prescribed perofmance functions}
    Let $\rho_i (t)$, $i \in \{1, \dots ,n\}$, be prescribed performance functions according to Definition~\ref{Definition of Prescribed performance function}, and let $c_i \geq 0$ be constant scalars such that at least one $c_i>0$. Then, $\varrho(t) :=\sum^n_{i = 1} c_i \rho_i (t)$ and $\varphi(t):= \sqrt{\sum_{i=1}^{n} c_i\, \rho_i(t)^2}$ are prescribed performance functions.
\end{lemma}
\begin{proof}
    Consider first the weighted sum $\varrho(t)$. Since $\rho_i(t) \geq \rho_{i,L}>0$ and at least one $c_i>0$, it follows that $\varrho(t) \geq \varrho_L := \sum_{i=1}^{n}c_i\rho_{i,L}>0$. Moreover, (i) Since $\rho_i \in \mathcal{C}^1$ for all $i$, $\sum^n_{i = 1} c_i \rho_i (t) \in \mathcal{C}^1$ holds; (ii)-(iii) By Definition~\ref{Definition of Prescribed performance function}, $\rho_i (t) \leq \rho_{i,U}$ and $|\dot{\rho}_i(t)|\leq \rho^d_{i,U}$ hold for some finite $\rho_{i,U} > 0$ and $\rho^d_{i,U}> 0$. Hence, $\varrho(t) \leq \varrho_U := \sum^n_{i = 1} c_i\rho_{i,U}$ and $|\sum^n_{i = 1}c_i \dot{\rho}_i(t)| \leq \sum^n_{i = 1} c_i \rho^d_{i,U}$ hold. 
        
    For $\varphi(t)$, since $\rho_i(t)\geq\rho_{i,L}>0$ and at least one $c_i>0$, it follows that $\varphi(t) \geq \varphi_{L} := \sqrt{\sum_{i=1}^{n}c_i\rho_{i,L}^2}>0$ for all $t\geq0$. Moreover: (i) since $\sum_{i=1}^{n} c_i\, \rho_i(t)^2 \in \mathcal{C}^1$, $\sum_i c_i \rho_i(t)^2 > 0$, and the square-root function is $\mathcal{C}^1$ on $\mathbb{R}_{> 0}$, $\varphi(t) \in \mathcal{C}^1$; (ii) $\varphi(t) \leq \varphi_U := \sqrt{\sum^n_{i=1} c_i \rho_{i,U}^2} < \infty$ holds from $\rho_i(t) \leq \rho_{i,U}$; (iii) $\dot \varphi(t) = \frac{\sum_{i=1}^n c_i \rho_i(t)\dot\rho_i(t)}{\varphi(t)}$ holds by the chain rule. Hence, since $\rho_{i,L} \leq \rho_i(t) \leq \rho_{i,U}$, $|\dot\rho_i(t)| \leq \rho^d_{i,U}$, and $\varphi(t)$ is uniformly bounded away from zero, $|\dot \varphi(t)|$ is bounded. 
\end{proof}

\subsection{Worst-case prescribed performance bounds}\label{Section: worst-case design of prescribed performance bounds}
Since $M_i$ is time-varying, condition \eqref{eq: constraint on the prescribed prformance functions} must hold for every $t \geq 0$. 
The matrix $M_i$ varies only through the switching signal $\sigma(t)$ and is constant for each graph $\mathcal{G}_l \in \mathbb{G}$. Therefore, a sufficient condition is obtained by replacing the time-varying $\sigma_{\min}(M_i)$ with its minimum over the finite set $\mathbb{G}$, which provides a uniform lower bound for all $\mathcal{G}_l \in \mathbb{G}$.
Hence, \eqref{eq: constraint on the prescribed prformance functions} is satisfied for all $t \geq 0$ and $j \in \Outhop{i}{k}$ if $\norm{\vect{\rho}_i(t)} \leq \min_{\mathcal{G}_l \in \mathbb{G}} \{\sigma_{\min}(M_i)\} \min_{j \in \Outhop{i}{k}} \{\delta_i^{j}(t)\}$.

When the set of admissible communication topologies $\mathbb{G}$ is known, the above condition can be evaluated directly. However, in many applications such information is unavailable or difficult to characterize. Thus, the following proposition establishes topology-independent worst-case bounds on the singular values of $M_i$ that are valid for every admissible communication graph with $N$ agents.
\begin{proposition}
    \label{Proposition: proof of proposition on norm(M_i^-1)}
    Irrespective of the set of admissible communication topologies $\mathbb{G}$, $M_i$ satisfies $\sigma_{\max}(M_i) = \sigma_{\min}(M_i) = 1$ if $N = 3$, and, for $N \geq 4$,
    \begin{equation}
        \label{eq: upper and lower bound on the singular values of M_i}
        \begin{split}
            \text{(i) }  &\sigma_{\max}(M_i) \leq 2N -5, \\
            \text{(ii) } &\sigma_{\min}(M_i) \geq \left(\frac{N-3}{(N-2)(N^2 - 3N+1)}\right)^{(N-3)/2},
        \end{split}
    \end{equation}
    where $N$ denotes the total number of agents.
\end{proposition}
\begin{proof}
    Given in Appendix~\ref{Appendix: proof of proposition on norm(M_i^-1)}.
\end{proof}
Using the lower bound on $\sigma_{\min}(M_i) $, condition \eqref{eq: constraint on the prescribed prformance functions} is guaranteed for every communication topology provided that
\begin{equation}
    \label{Eq: rho design bound infinite-norm, case random graph}
    \|\vect{\rho}_i(t)\| \leq \gamma_i(N) \min_{j \in \Outhop{i}{k}} \{\delta_i^{j}(t)\},
\end{equation}
where
\begin{equation*}
    \gamma_i(N) := \begin{cases}
        \min_{\mathcal{G}_l \in \mathbb{G}}\{\sigma_{\min}(M_i)\}, & \mathbb{G} \text{ known},\\
        1, & N = 3,\\
        \left(\tfrac{N-3}{(N-2)(N^2-3N+1)}\right)^{\frac{N-3}{2}}, & N \geq 4,
    \end{cases}
\end{equation*}
with the last two cases applying when $\mathbb{G}$ is unknown.

Then, Problem~\ref{Problem state estimation error convergence} can be reformulated as:
\begin{problem}
    Given the heterogeneous MAS \eqref{eq: agent's dynamic} communicating over the directed switching graphs $\mathcal{G}_{\sigma(t)} \in \mathcal{G}$, and prescribed performance functions $\delta^{j}_i(t)$, design a decentralized \mbox{$k$-hop} observer such that $|\xi_i^{j}(t)| < \rho_i^{j}(t)$ holds for all $t \geq 0$, $i \in \mathcal{V}$ and $j \in \Outhop{i}{k}$, where the prescribed performance functions $\rho_i^{j}(t)$ satisfy \eqref{Eq: rho design bound infinite-norm, case random graph}.
\end{problem}

\begin{remark}
    Since all $\rho_i^{j}(t)$ and $\delta_i^{j}(t)$ are design choices, we can impose desired behavior on every $\error{x}{j}{i}(t)$.
\end{remark}

For $n_i > 1$, the same reasoning applies componentwise to $\error{x}{j}{i}(t)$. Specifically, the prescribed performance can be enforced by constraining each component $[\xi_{i}^{j}]_l$, $l \in \{1, \dots, n_i\}$, of $\xi_{i}^{j}$.

\section{Prescribed performance observer under switching topologies}
\label{Sec: Decentralized PPSO}
As discussed in Section~\ref{Section: state and input estimate}, changes in the dimension of the state estimation error vector $\error{\vect{x}}{}{i}$ may introduce discontinuous jumps at switching instants. Consequently, the estimation-error dynamics become hybrid, requiring the estimation errors to be analyzed both during continuous evolution and across switching instants. This section establishes conditions ensuring that the estimation errors preserve the desired prescribed performance guarantees despite the switching communication topology.

\subsection{Prescribed-performance function design}
\label{Section: Prescribed-performance funnel design}
Although the prescribed performance functions introduced in Section~\ref{Problem formulation} guarantee the desired transient and steady-state behavior during continuous evolution, switching instants may introduce discontinuities in the estimation errors. In particular, when the $k$-hop neighborhood of agent $i$ changes, newly included estimates can increase the disagreement terms $\xi_i^{j}$ defined
in \eqref{eq: disagreement vector expression}, thereby violating the prescribed performance constraints in
\eqref{eq: definition of the bounds on the disagreement term} immediately after a switch.

To address this issue, we develop a systematic procedure to construct auxiliary prescribed-performance functions that ensure the feasibility of the prescribed-performance constraints after each topology switch. Specifically, we replace the original performance bounds with lower envelopes that account for the worst-case variation of the active \mbox{$k$-hop} neighborhood.

To this end, we first show that the pointwise minimum of exponential prescribed performance functions admits an exponential lower approximation.
\begin{proposition}
    \label{Proposition: lower bound on the minimum prescribe performance}
    For each $j \in \mathcal{O}_i^{\cup}$, let
    \begin{equation*}
        \delta_i^j(t):=(\delta_i^j(0)-\delta_i^j(\infty))e^{-l_i^jt} +\delta_i^j(\infty),
    \end{equation*}
    where $\delta_i^j(0)>\delta_i^j(\infty)>0$, and $l_i^j>0$. Define $a_i := \min_{j\in\mathcal{O}_i^{\cup}} (\delta_i^j(0)-\delta_i^j(\infty))$, $b_i := \min_{j\in\mathcal{O}_i^{\cup}}\delta_i^j(\infty)$ and $\bar l_i := \max_{j\in\mathcal{O}_i^{\cup}} l_i^j$, and let $\bar\delta_i(t) := a_ie^{-\bar l_it}+b_i$. Then $\bar\delta_i(t)$ is a prescribed performance function such that $\bar\delta_i(t) \le \min_{j\in\Outhop{i}{k}(t)} \delta_i^j(t)$ for all $t\ge0$.
\end{proposition}
\begin{proof}
    Given in Appendix~\ref{Appendix: proof of proposition lower bound on the minimum prescribe performance}.
\end{proof}

Proposition~\ref{Proposition: lower bound on the minimum prescribe performance} provides an exponential lower bound $\bar\delta_i(t) := a_ie^{-\bar l_it}+b_i$ for $\min_{j\in\Outhop{i}{k}} \delta_i^j(t)$. Since $\vect{\rho}_i(t)$ has $\eta_i$ components and $\eta_i \leq N-2$, its Euclidean norm satisfies $\|\vect{\rho}_i(t)\| \leq \sqrt{\eta_i}\|\vect{\rho}_i(t)\|_\infty \leq \sqrt{N-2} \|\vect{\rho}_i(t)\|_\infty$. Therefore, condition \eqref{Eq: rho design bound infinite-norm, case random graph} can be enforced by selecting 
\begin{equation}
    \label{eq: design of rho_i based on the minimum delta}
    \rho_i^j(t) = \rho_i(t):=\frac{\gamma_i(N)}{\sqrt{N-2}}\bar\delta_i(t)
\end{equation}
for all $j\in\Outhop{i}{k}$ and $t\geq0$. With this choice, the prescribed performance functions $\rho_i^j(t)$  retain the exponential form of the lower bound $\bar\delta_i(t)$, i.e., $\rho_i^j(t) = (\rho^j_i(0)- \rho^j_i(\infty))e^{-\bar l_i t}+\rho^j_i(\infty)$.

To cope with the hybrid nature of the problem while preserving the prescribed performance constraints in \eqref{eq: definition of the bounds on the disagreement term}, we introduce a family of prescribed performance functions that are updated at each switching instant. Specifically, for each interval $[t_{\tau}, t_{\tau+1})$, we define $\underline{\rho}^{j}_i(t, t_{\tau}) := (\underline{\rho}^{j}_{i,\tau}(t_\tau) - \underline{\rho}^{j}_{i, \tau}(\infty)) e^{- \underline{l}^{j}_{i,\tau} (t- t_{\tau})} + \underline{\rho}^{j}_{i, \tau}(\infty)$, with $\underline{\rho}^{j}_{i,\tau}(t_\tau) = \rho^{j}_i(t_{\tau})$ from \eqref{eq: definition of the bounds on the disagreement term}, $\underline{l}^{j}_{i, \tau} \geq \overline{l}_{i}$ and $0 < \underline{\rho}^{j}_{i, \tau}(\infty) \leq {\rho}^{j}_{i} (\infty)$. The following proposition establishes a key property of $\underline{\rho}^{j}_i(t, t_{\tau})$.
\begin{proposition}
    \label{Proposition: underline rho lower than rho}
    Let $\rho^j_i(t)$ and $\underline{\rho}^{j}_i(t, t_{\tau})$ be defined as above. Then, for every $t_\tau$ and all $t \geq t_\tau$,
    $\underline{\rho}^{j}_i(t, t_{\tau}) \leq \rho^j_i(t)$.
\end{proposition}
\begin{proof}
    Given in Appendix~\ref{Appendix: proof of underline rho lower than rho}.
\end{proof}
\begin{figure}[t!]
    \centering
    \includegraphics[width=1\linewidth]{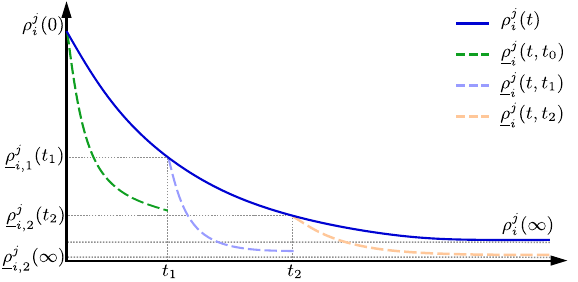}
    \caption{Prescribed performance funnels design.}
    \label{fig: Funnels design}
\end{figure}
Fig.~\ref{fig: Funnels design} illustrates an example of the prescribed performance funnels resulting from the proposed design.
\begin{remark}
    \label{remark: remark on equal rho for all k-hop neighbors}
    Since $\rho_i^j(t)=\rho_i(t)$ by design, the updated prescribed performance functions can be chosen identically for all $k$-hop out-neighbors, i.e., $\underline{\rho}^{j}_i(t, t_{\tau}) = \underline{\rho}_i(t, t_{\tau})$ for all $j \in \Outhop{i}{k}$.
\end{remark}

Proposition~\ref{Proposition: underline rho lower than rho} establishes that $\underline{\rho}^{j}_i(t, t_{\tau}) \leq \rho^{j}_i(t)$ for all $t_\tau$ and $t \geq t_\tau$. Thus, enforcing
\begin{equation}
    \label{eq: final constraint to be satisfied by the disagreement vector}
    |\xi^j_i| < \underline{\rho}^{j}_i(t, t_{\tau}) 
\end{equation}
for all $t\in [t_{\tau}, t_{\tau+1})$ and $j \in \Outhop{i}{k}$ guarantees that the prescribed performance constraints \eqref{eq: definition of the bounds on the disagreement term} are satisfied. The remaining challenge is to determine the condition under which this inequality is satisfied immediately after each topology switch. The following subsection derives sufficient conditions for this purpose.
Given Remark~\ref{remark: remark on equal rho for all k-hop neighbors}, we replace $\rho_i^j(t)$ and $\underline{\rho}^{j}_i(t, t_{\tau})$ with $\rho_i(t)$ and $\underline{\rho}_i(t, t_{\tau})$ in the following sections.

\subsection{Conditions for prescribed performance preservation}
To guarantee that \eqref{eq: final constraint to be satisfied by the disagreement vector}, and thus \eqref{eq: definition of the bounds on the disagreement term}, is satisfied after every topology switch, the following should hold for all $i \in \mathcal{V}$ and $t_\tau$:
\begin{equation}
    \label{eq: general condition to be satisfied at the stime of switch}
    |\vect{\xi}_i(t_{\tau}^+)| < \rho_i(t_{\tau}) 1_{\eta_i}.
\end{equation}

From \eqref{eq: State estimation error after the switch} and \eqref{eq: disagreement vector expression}, the disagreement vector $\vect{\xi}_i(t_{\tau}^+)$ satisfies $\vect{\xi}_i(t_{\tau}^+) = M_i(t_{\tau}^+) \Xi_i^\tau  \error{\vect{x}}{}{i}(t_{\tau}^-) + M_i(t_{\tau}^+) \phi_i(t_{\tau})$. Substituting $\error{\vect{x}}{}{i}(t_{\tau}^-) = M_i^{-1}(t_{\tau}^-)\vect{\xi}_i(t_{\tau}^-)$ gives
\begin{equation}
    \label{eq: original condition that should hold at the switching instant}
    \hspace{-0.25cm}
     |M_i(t_{\tau}^+) \Xi_i^\tau M_i^{-1}(t_{\tau}^-)\vect{\xi}_i(t_{\tau}^-) +  M_i(t_{\tau}^+) \phi_i(t_{\tau})| < \rho_i(t_\tau) 1_{\eta_i}.
\end{equation}

We next derive uniform bounds for the two terms on the left-hand side.
Similarly to \eqref{Eq: rho design bound infinite-norm, case random graph}, define
\begin{equation}
    \label{eq: alpha(N) definition}
    \alpha_i(N) := \begin{cases}
        \max_{\mathcal{G}_{l} \in \mathbb{G}} \{\sigma_{\max}(M_i)\}, &\text{ if $\mathbb{G}$ is known,} \\
        2N-5, &\text{ otherwise,}
    \end{cases}
\end{equation}
which provides a uniform upper bound on  $\norm{M_i}$ over all admissible communication graphs.

For the switching-induced term, assume that $\norm{\phi_i(t_{\tau})} \leq \phi_i^{\max}$, where $\phi_i(t_\tau)$ collects the estimation errors associated with the newly added $k$-hop out-neighbors. As discussed in Remark~\ref{Remark: phi can be made arbitrarily small}, $\phi_i^{\max}$ can be made arbitrarily small through appropriate initialization and design of the prescribed performance functions. Under these bounds, the following
lemma provides a sufficient \mbox{dwell-time} condition guaranteeing preservation
of the prescribed performance constraints across topology switches.

\begin{lemma}
    \label{Lemma: prescribed performance preservation at switching instants}
     Let $\gamma_i(N)$ and $\alpha_i(N)$ be defined as in \eqref{Eq: rho design bound infinite-norm, case random graph} and \eqref{eq: alpha(N) definition}, respectively, and define $\theta_i(N):=\sqrt{N-2}\gamma_i^{-1}(N)\alpha_i(N)$ together with 
    \begin{equation*}
        \omega_{i,\tau}(t,N):= \rho_i(t) -\alpha_i(N)\phi_i^{\max} -\theta_i(N)\underline{\rho}_{i,\tau-1}(\infty).
    \end{equation*}
    For every $i\in \mathcal{V}$, let the prescribed-performance function $\underline{\rho}_i(t, t_{\tau-1})$ be designed as in Section~\ref{Section: Prescribed-performance funnel design}. Assume $\norm{\phi_i(t_{\tau})} \leq \phi_i^{\max}$, and that
    \begin{equation}
        \label{eq: asymptotic feasibility condition}
        \omega_{i,\tau}(\infty,N):= \lim_{t\rightarrow \infty} \omega_{i,\tau}(t,N) > 0
    \end{equation}    
    is satisfied by appropriate design of $\rho_i(t)$ and $\underline{\rho}_{i,\tau-1}(t)$. For each agent $i$, define the function
    \small
    \begin{equation}
        \label{eq: def of F(T)}
        F_{i,\tau}(T) := \frac{1}{\underline l_{i,\tau-1}}\ln\left(\frac{\theta_i(N)\big(\rho_i(t_{\tau-1})- \underline\rho_{i,\tau-1}(\infty)\big)}{\omega_{i,\tau}(t_{\tau-1}+T,N)}\right)
    \end{equation}
    \normalsize
    \noindent and the critical dwell-time threshold $T_{i,\tau}^{\min} := \inf\big\{T\geq0 : T' > F_{i,\tau}(T') \;\;\forall\, T'\geq T \big\}$. Then, if the
    dwell time $T_\tau := t_\tau - t_{\tau-1}$ satisfies
    \begin{equation}
        \label{eq: dwell time condition}
        T_\tau > \max_{i\in \mathcal{V}}\{T_{i,\tau}^{\min}\},
    \end{equation}
    and \eqref{eq: final constraint to be satisfied by the disagreement vector} holds for every $i\in \mathcal{V}$ and $t \in [t_{\tau-1}, t_{\tau})$, condition \eqref{eq: general condition to be satisfied at the stime of switch} is satisfied for every $i$.
\end{lemma}
\begin{proof}
    Given in Appendix~\ref{Appendix: proof of lemma prescribed performance preservation at switching instants}.
\end{proof}

Condition~\eqref{eq: asymptotic feasibility condition} requires that the steady-state prescribed performance margin remains strictly positive after accounting for the worst-case switching-induced perturbation. Consequently, it is a sufficient condition for the existence of a finite dwell time satisfying~\eqref{eq: dwell time condition}. Topology transitions violating this condition are not certified by Lemma~\ref{Lemma: prescribed performance preservation at switching instants} and are therefore not guaranteed to preserve the prescribed performance.

\begin{remark}
    \label{remark: well-poseness of F}
    The function $F_{i,\tau}(T)$ in \eqref{eq: def of F(T)} is well defined for all $T\geq0$. Indeed, since $\rho_i(t)$ decreases exponentially to $\rho_i(\infty)$, and $\omega_{i,\tau}(t,N)$ differs from $\rho_i(t)$ only by a constant $\alpha_i(N)\phi_i^{\max} +\theta_i(N)\underline{\rho}_{i,\tau-1}(\infty)$, \eqref{eq: asymptotic feasibility condition} implies $\omega_{i,\tau}(t,N)\geq\omega_{i,\tau}(\infty,N)>0$ for all $t \geq t_{\tau-1}$. Moreover, Proposition~\ref{Proposition: underline rho lower than rho} together with the design of $\rho_i(t)$ gives $\rho_i(t_{\tau-1})>\underline{\rho}_{i,\tau-1}(\infty)$. Hence, since $\theta_i(N) >0$ and $\underline{l}_{i,\tau-1}> 0$, the argument of the logarithm in \eqref{eq: def of F(T)} is positive, and $F_{i,\tau}(T)$ is well defined.
\end{remark}
\begin{remark}
    \label{Remark: phi can be made arbitrarily small}
    Since the condition $\norm{\phi_i(t_{\tau})} \leq \phi_i^{\max}$ constrains the norm of the estimation errors of the newly added $k$-hop out-neighbors, condition~\eqref{eq: asymptotic feasibility condition} may appear restrictive. However, since $\underline{\rho}_{i,\tau-1}(\infty)$ and $\rho_i(t)$ are design parameters, it can in general be satisfied through appropriate design. In particular, suppose that, for every newly added $k$-hop neighbor $j \in \Outhop{i}{k}(t_{\tau}) \setminus \Outhop{i}{k}(t_{\tau-1})$, there exists an agent $l$ such that either (i) $l \in \inN{j}(t_{\tau}) \cap \outN{i}(t_\tau)$, or (ii) $ l \in \inN{j}(t_{\tau}) \cap \Outhop{i}{k}(t_{\tau-1})$. In case (i), $l$ has access to $x_i$ by Assumption~\ref{Assumption on 2-hop communication}, so that initializing $\hat{x}^j_i = x_i$ yields $\tilde{x}^j_i(t_\tau) = 0$. In case (ii), initializing $\hat{x}^j_i = \hat{x}^l_i$ and choosing $\underline{\rho}_{i,\tau-1}(\infty)$ sufficiently small makes the initialization error arbitrarily small.

    Indeed, if $\underline{\rho}_{i,\tau-1}(\infty) \simeq 0$ and \eqref{eq: final constraint to be satisfied by the disagreement vector} holds in $[t_{\tau-1}, t_\tau)$, then \eqref{eq: disagreement vector expression} and Lemma~\ref{Lemma: -M_i is Hirwitz} imply $\error{x}{l}{i} \simeq 0$ for all $l\in\Outhop{i}{k}(t_{\tau-1})$. Since both $j$ and $l$ estimate $x_i$, the initialization above yields $\error{x}{j}{i}(t_\tau)=\error{x}{l}{i}(t_\tau)\simeq0$. Hence, $\norm{\phi_i(t_\tau)}$ can be made arbitrarily small, and the condition $\norm{\phi_i(t_\tau)}\leq\phi_i^{\max}$ can be satisfied for any desired $\phi_i^{\max}>0$.
\end{remark}

Lemma~\ref{Lemma: prescribed performance preservation at switching instants} guarantees the preservation of the prescribed-performance constraints across topology switches. Section~\ref{Section: observer design} presents the observer design that ensures satisfaction of these constraints between consecutive switching instants.

\subsection{Prescribed performance observer design}
\label{Section: observer design}
The objective of the observer is to ensure that the disagreement vector $\vect{\xi}_i$ satisfies \eqref{eq: final constraint to be satisfied by the disagreement vector}. For each time interval $[t_{\tau}, t_{\tau +1})$, define $e_i^{j}(t, t_\tau) \in (-1,1)$ as the normalization of $\xi_{i}^{j}(t)$ with respect to $\underline{\rho}_i(t, t_\tau)$, i.e., $e_i^{j}(t, t_\tau) :=\underline{\rho}_i(t, t_\tau)^{-1}\xi_{i}^{j}(t)$, and let $\Upsilon: (-1,1) \rightarrow \mathbb{R}$ be a strictly increasing transformation satisfying $\Upsilon(0) = 0$.
Following the PPC framework \cite{BechlioulisPPC}, define the transformed normalized disagreement for all $i \in \mathcal{V}$ and $j\in \Outhop{i}{k}$ as:
\begin{equation}
    \label{eq: transformed normalized disagreement}
    \epsilon^{j}_i(t,t_\tau) := \Upsilon(e^{j}_i(t, t_{\tau})) = \Upsilon(\underline{\rho}_i(t, t_\tau)^{-1}\xi_{i}^{j}).
\end{equation}
In this work, we select $\Upsilon(e) = \ln(\frac{1+e}{1-e})$, whose derivative is strictly positive and given by $ J_{\Upsilon}(e) = \frac{2}{1-e^2}$. 

\begin{remark}
    \label{remark: t_tau as aparameter wtithin the time window}
    For $t\in[t_\tau,t_{\tau+1})$, the switching instant $t_\tau$ is fixed and treated as a parameter in $\epsilon_i^j(t,t_\tau)$. Accordingly, all derivatives are taken with respect to $t$.
\end{remark}
Based on \eqref{eq: transformed normalized disagreement}, the proposed $k$-hop prescribed performance state observer ($k$-hop PPSO) is defined as:
\begin{equation}
    \label{Eq: State observer dynamics in components}
    \dot{\hat{x}}_{i}^{j} = f_i(\estimate{x}{j}{i}) - \underline{\rho}_i(t, t_{\tau})^{-1} J_{\Upsilon}(e_i^{j}) \epsilon_{i}^{j}(t, t_{\tau}).
\end{equation}
Stacking \eqref{Eq: State observer dynamics in components} for all $j \in \Outhop{i}{k}$ yields
\begin{equation}
    \label{Eq: state estimation dynamics in vector form}
    \dot{\hat{\vect{x}}}_i =\vect{f}_i(\estimate{\vect{x}}{}{i}) -  \underline{\vect{P}}_i^{-1} \vect{J}_i \vect{\epsilon}_i,
\end{equation}
where $\hat{\vect{x}}_i$ is as per \eqref{eq: estimation of state and input of agent i}, $\vect{f}_i(\estimate{\vect{x}}{}{i}) := [\state{f}{}{i}(\estimate{x}{j}{i})]_{j \in \Outhop{i}{k}}$, $\underline{\vect{P}}_i:= \underline{\rho}_i(t, t_\tau) I_{\eta_i}$, $\vect{J}_i := \text{diag}(J_{\Upsilon}(e^{j}_i))_{j \in \Outhop{i}{k}}$, and $\vect{\epsilon}_{i}  := [\epsilon^{j}_i]_{j \in \Outhop{i}{k}}$.

The following theorem establishes that the proposed k-hop PPSO guarantees satisfaction of \eqref{eq: final constraint to be satisfied by the disagreement vector}.

\begin{theorem}
    \label{Theorem: main theorem on state estimation convergence in time varying case}
    Consider the heterogeneous MAS \eqref{eq: agent's dynamic} communicating over the directed switching graph $\mathcal{G}_{\sigma(t)} \in \mathbb{G}$, equipped with the decentralized state observers \eqref{Eq: State observer dynamics in components} and prescribed performance functions $\delta^{j}_i(t)$. Assume that, for every $i \in \mathcal{V}$ and $t_{\tau}$, $\norm{\phi_i(t_{\tau})} \leq \phi_i^{\max}$, with $\phi_i^{\max} \in \mathbb{R}_{\geq 0}$. Under \mbox{Assumptions~\ref{Assumption on existence of a solution}-\ref{Assumption on 2-hop communication}}, if the prescribed performance functions are designed as in Section~\ref{Section: Prescribed-performance funnel design}, the dwell times satisfy \eqref{eq: dwell time condition}, and the disagreement vector is initialized such that $|\xi_i^j(0)|<\rho_i(0)$, then $|\tilde{x}_{i}^{j}(t)| < \delta^{j}_i(t)$ holds for all $t \geq 0$, $i \in \mathcal{V}$ and $j \in \Outhop{i}{k}(t)$.
\end{theorem}
\begin{proof}
    We proceed by induction over the switching intervals. For the induction step,
    consider an arbitrary interval $[t_{\tau},t_{\tau+1})$ and assume that the
    prescribed-performance condition holds on all preceding switching intervals.
    We then show that it remains satisfied for all
    $t\in[t_{\tau},t_{\tau+1})$. Since the argument is independent of the
    switching index $\tau$, the same reasoning applies to every switching
    interval. The induction is initialized at $\tau=0$ by the stated
    prescribed-performance condition
    $|\xi_i^j(0)|<\rho_i(0)$.
    
    Fix $i \in \mathcal{V}$. By Remark~\ref{remark: t_tau as aparameter wtithin the time window}, $t_\tau$ is fixed throughout the interval and is therefore treated as a parameter. Hence, all derivatives are taken with respect to $t$. Differentiating $\vect{\epsilon}_{i}$ yields $\dot{\vect{\epsilon}}_i = \vect{J}_i  \underline{\vect{P}}^{-1}_{i}(\dot{\vect{\xi}}_{i} - \dot{\underline{\vect{P}}}_{i} \vect{e}_i)$. Since $\sigma(t)$ is constant for all $t \in [t_{\tau}, t_{\tau+1})$, all graph-dependent quantities are constant as well. Therefore, from \eqref{eq: disagreement vector expression},
    \begin{equation}
        \label{eq: transformed error dynamics in time varying case}
        \dot{\vect{\epsilon}}_i = \vect{J}_i \underline{\vect{P}}_i^{-1}(M_i \dot{\tilde{\vect{x}}}_{i} - \dot{\underline{\vect{P}}_{i} }\vect{e}_i).
    \end{equation}
    From \eqref{eq: agent's dynamic}, \eqref{Definition: error on input and state estimation of agent i} and \eqref{Eq: State observer dynamics in components}, the state estimation error for $t \in [t_{\tau},t_{\tau+1})$ satisfies
    \begin{equation}
        \label{eq: state estimation erro dynamics in time varying case}
        \dot{\tilde{\vect{x}}}_{i} = \vect{f}_i(\estimate{\vect{x}}{}{i}) - \vect{f}_i(\vect{x}_{i}) - \vect{g}_i -  \underline{\vect{P}}_i^{-1} \vect{J}_i \vect{\epsilon}_i - \vect{w}_i,
    \end{equation}
    with $\vect{f}_i(\vect{x}_i) := 1_{\eta_i} \otimes f_i(x_i)$, $\vect{g}_i := 1_{\eta_i} \otimes g_i(u_i) $ and $\vect{w}_i := 1_{\eta_i} \otimes w_i(\vect{x},t)$.
    
    Before proceeding with the Lyapunov analysis, we first establish the existence of a local solution to the state estimation error dynamics. To this end, let $\tilde{\mathcal{X}}_i(t, t_\tau) := \{ \error{\vect{x}}{}{i} \in \mathbb{R}^{\eta_i}| -1_{\eta_i} <\vect{e}_i = \underline{\vect{P}}_i^{ -1} \vect{\xi}_i <1_{\eta_i}\}$ denote the set of state estimation errors for which the disagreement vector satisfies the prescribed performance bounds at time $t \in [t_{\tau}, t_{\tau+1})$.
    By Assumption~\ref{Assumption on existence of a solution}, $f_i$ is locally Lipschitz continuous, $g_i$ is measurable and essentially locally bounded, and $w_i$ is continuous and uniformly bounded. Furthermore, $\underline{\rho}_i(t, t_{\tau})$ is $\mathcal{C}^1$ by definition, and $J_{\Upsilon}$ and $\Upsilon$ are smooth on $(-1,1)$. Then, since by Assumptions~\ref{Assumption on existence of a solution}--\ref{Assumption on g} the agent trajectories $x_i(t)$ exist on $[t_\tau, t_{\tau+1})$, the right-hand side of \eqref{eq: state estimation erro dynamics in time varying case} is measurable in $t$, continuous in $\tilde{\vect{x}}_i$, and locally bounded on the bounded open set $\tilde{\mathcal{X}}_i(t, t_\tau)$. Thus, a local Carath\'eodory solution exists in $\tilde{\mathcal{X}}_i(t, t_\tau)$ for all $t \in [t_\tau, t_\tau + h)$, with $h > 0$ \cite[Proposition 1]{Cortes_DiscontinuousDS}. Since $\tilde{\vect{x}}_i$ is absolutely continuous, so is $V$ along solutions, and all derivative inequalities below hold for almost all $t$.

    As introduced in Section~\ref{Section: disagreement vector}, $M_i$ is a nonsingular \mbox{$M$-matrix}. Thus, according to Corollary~\ref{Corollary: existence of a diagonal matrix P that render an M-matrix Hurwitz}, there exists a diagonal positive definite matrix $A_i$ such that $A_i M_i + M_i^\top A_i \succ 0$. For the time interval $[t_{\tau}, t_{\tau+1})$, consider the candidate Lyapunov function $V = \vect{\epsilon}_i^\top A_i \vect{\epsilon}_i$. Substituting \eqref{eq: transformed error dynamics in time varying case} and \eqref{eq: state estimation erro dynamics in time varying case} into $\dot{V} = \vect{\epsilon}_i^\top A_i \dot{\vect{\epsilon}}_i + \dot{\vect{\epsilon}}_i^{\top} A_i \vect{\epsilon}_i$ yields:
     \begin{equation}
        \label{eq: Lyapunov function derivative, original}
        \begin{split}
            \dot{V} =& \vect{\epsilon}_i^\top A_i \vect{J}_i \underline{\vect{P}}_i^{-1}\Big\{M_i[ \vect{f}_i(\estimate{\vect{x}}{}{i}) - \vect{f}_i(\vect{x}_{i}) -  \underline{\vect{P}}_i^{-1} \vect{J}_i \vect{\epsilon}_i \\ & - \vect{g}_i - \vect{w}_i] - \dot{\underline{\vect{P}}_{i} } \vect{e}_i \Big\} + \Big\{M_i[ \vect{f}_i(\estimate{\vect{x}}{}{i}) - \vect{f}_i(\vect{x}_{i})  \\ & - \vect{g}_i -  \underline{\vect{P}}_i^{-1} \vect{J}_i \vect{\epsilon}_i - \vect{w}_i] - \dot{\underline{\vect{P}}_{i} } \vect{e}_i \Big\}^\top \underline{\vect{P}}_i^{-1} \vect{J}_i A_i \vect{\epsilon}_i.
        \end{split}
    \end{equation} 
    Since $A_i$, $\vect{J}_i$, and $\underline{\vect{P}}_i$ are diagonal, they are symmetric and mutually commuting. Thus, expanding \eqref{eq: Lyapunov function derivative, original} and rearranging the resulting terms yields:
    \begin{equation}
        \label{eq: Lyapunov function derivative, pre bound}
        \begin{split}
            \hspace{-0.2cm}\dot{V} = &2\vect{\epsilon}_i^\top \vect{J}_i \underline{\vect{P}}_i^{-1} A_i\Big\{ M_i[ \vect{f}_i(\estimate{\vect{x}}{}{i}) - \vect{f}_i(\vect{x}_{i}) - \vect{g}_i - \vect{w}_i] \\& -\dot{\underline{\vect{P}}_{i} } \vect{e}_i \Big\} - \vect{\epsilon}_i^\top \vect{J}_i \underline{\vect{P}}_i^{-1} (A_i M_i + M_i^\top A_i )\underline{\vect{P}}_i^{-1} \vect{J}_i \vect{\epsilon}_i.
        \end{split}
    \end{equation}

    Since $A_i M_i + M_i^\top A_i \succ 0$, the quadratic term in
    \eqref{eq: Lyapunov function derivative, pre bound} satisfies
    $- \vect{\epsilon}_i^\top \vect{J}_i \underline{\vect{P}}_i^{-1} (A_i M_i + M_i^\top A_i )\underline{\vect{P}}_i^{-1} \vect{J}_i \vect{\epsilon}_i \leq -\Bar{\lambda}_{\min} \norm{\underline{\vect{P}}_i^{-1} \vect{J}_i \vect{\epsilon}_{i}}^2$,
    with $\Bar{\lambda}_{\min} :=\lambda_{\min} (A_i M_i + M_i^\top A_i) > 0$. Hence,
    \eqref{eq: Lyapunov function derivative, pre bound} admits the upper bound
    \begin{equation}
        \label{eq: Lyapunov function derivative, quadratic bound}
        \dot{V} \leq -\Bar{\lambda}_{\min} \norm{\underline{\vect{P}}_i^{-1} \vect{J}_i \vect{\epsilon}_{i}}^2 + \vect{\epsilon}_i^\top \vect{J}_i \underline{\vect{P}}_i^{-1} \vect{b}(t),
    \end{equation}
    where $\vect{b}(t) := 2 A_i M_i[ \vect{f}_i(\estimate{\vect{x}}{}{i}) - \vect{f}_i(\vect{x}_{i}) - \vect{g}_i - \vect{w}_i] - 2 A_i\dot{\underline{\vect{P}}_{i} } \vect{e}_i$.
    
    Fix $0 < \zeta < \Bar{\lambda}_{\min}$ and split the first term of
    \eqref{eq: Lyapunov function derivative, quadratic bound} as
    $-\Bar{\lambda}_{\min} \norm{\underline{\vect{P}}_i^{-1} \vect{J}_i \vect{\epsilon}_{i}}^2 = -(\Bar{\lambda}_{\min} - \zeta) \norm{\underline{\vect{P}}_i^{-1} \vect{J}_i \vect{\epsilon}_{i}}^2 - \zeta \norm{\underline{\vect{P}}_i^{-1} \vect{J}_i \vect{\epsilon}_{i}}^2$.
    Since $\vect{J}_i$ and $\underline{\vect{P}}_i$ are diagonal,
    $\norm{\underline{\vect{P}}_i^{-1} \vect{J}_i \vect{\epsilon}_{i}}^2 \geq \alpha_{J}\alpha_{\rho} \norm{\vect{\epsilon}_i}^2$,
    with $\alpha_{J} :=\min_{e \in (-1,1)} J_{\Upsilon}(e)^2= 4$ and
    $\alpha_{\rho} := (\sup_{t \in [t_\tau, t_{\tau+1})} \underline{\rho}_i(t, t_{\tau})^2)^{-1}$.
    By Definition~\ref{Definition of Prescribed performance function}, each prescribed performance function is bounded, so $\alpha_{\rho} > 0$. Thus, $\dot{V} \leq -(\Bar{\lambda}_{\min} - \zeta)\alpha_{J} \alpha_{\rho} \norm{\vect{\epsilon}_i}^2 - \zeta \norm{\underline{\vect{P}}_i^{-1} \vect{J}_i \vect{\epsilon}_{i}}^2 + \vect{\epsilon}_i^\top \vect{J}_i \underline{\vect{P}}_i^{-1}\vect{b}(t).$

    Since $\underline{\vect{P}}_i^{-1} \vect{J}_i$ is symmetric,
    $\norm{\sqrt{\zeta}\, \underline{\vect{P}}_i^{-1} \vect{J}_i \vect{\epsilon}_{i} - \tfrac{1}{2\sqrt{\zeta}}\vect{b}(t)}^2\geq 0$ gives
    $\vect{\epsilon}_i^\top \vect{J}_i \underline{\vect{P}}_i^{-1}\vect{b}(t) - \zeta \norm{\underline{\vect{P}}_i^{-1} \vect{J}_i \vect{\epsilon}_{i}}^2 \leq \tfrac{1}{4 \zeta}\vect{b}^\top(t)\vect{b}(t)$,
    and
    \vspace{-0.2cm}
    \begin{equation}
        \vspace{-0.1em}
        \label{eq: Lyapunov function upper bound before introduction of k}
        \dot{V} \leq -(\Bar{\lambda}_{\min} - \zeta) \alpha_{J}\alpha_{\rho} \vect{\epsilon}_i^\top \vect{\epsilon}_i + \frac{1}{4 \zeta}\vect{b}^\top(t)\vect{b}(t).
    \end{equation}
    
    Since $A_i$ is diagonal and positive definite, $V = \vect{\epsilon}_i^\top A_i \vect{\epsilon}_i \leq \lambda_{\max}(A_i) \vect{\epsilon}_i^\top \vect{\epsilon}_i$ holds. Thus, $\vect{\epsilon}_i^\top \vect{\epsilon}_i \geq \lambda_{\max}(A_i)^{-1} V$, and \eqref{eq: Lyapunov function upper bound before introduction of k} can be rewritten as:
    \vspace{-0.3cm}
    \begin{equation}
         \vspace{-0.1em}
        \label{eq: final Lyapunov function upper bound}
        \dot{V} \leq - \kappa V + \beta(t),
    \end{equation}
    with $\kappa:= (\Bar{\lambda}_{\min} - \zeta) \alpha_{J}\alpha_{\rho} \lambda_{\max}(A_i)^{-1} $ and $\beta(t) := \frac{\lambda_{\max}(A_i)^2}{\zeta} \{\sigma_{\max}(M_i)[\norm{\vect{f}_i(\estimate{\vect{x}}{}{i})- \vect{f}_i(\vect{x}_{i})} + \norm{\vect{w}_i} + \norm{\vect{g}_i}] + \norm{\dot{\underline{\vect{P}}_{i} } \vect{e}_i}\}^2$.

    We next show that $\beta(t)$ admits a finite upper bound $\Bar{\beta}$. The inequality $|\error{x}{j}{i}(t)| <\delta^{j}_i(t)$, holding in $\tilde{\mathcal{X}}_i(t, t_\tau)$, implies that $\tilde{\mathcal{X}}_i(t, t_\tau) $ is a bounded open set. Therefore, Assumption~\ref{Assumption on g} guarantees that $\norm{\vect{f}_i(\estimate{\vect{x}}{}{i})- \vect{f}_i(\vect{x}_{i})}$ is bounded on $\tilde{\mathcal{X}}_i(t, t_\tau)$. Moreover, $\dot{\underline{\vect{P}}_{i} } \vect{e}_{i}$ is a column vector with entries $ \dot{\underline{\rho}}_i e^{j}_i$. Since $|e^{j}_i|< 1$ in $\tilde{\mathcal{X}}_i(t, t_\tau)$, and $|\dot{\underline{\rho}}_i|$ is finite by Definition \ref{Definition of Prescribed performance function}, $\norm{\dot{\underline{\vect{P}}_{i} } \vect{e}_i}$ is bounded. Together with the boundedness of $\norm{\vect{g}_i}$ and $\norm{{\vect{w}}_i}$ by Assumptions~\ref{Assumption on g}(ii) and~\ref{Assumption on existence of a solution}(iii), these bounds imply the existence of a finite upper bound $\Bar{\beta}$ satisfying $\beta(t) \leq \Bar{\beta} < \infty$ for all $ \error{\vect{x}}{}{i} \in \tilde{\mathcal{X}}_i(t, t_\tau)$.

    Inspired by \cite[Thm. 22]{SLiu_Contracts}, the invariance of $\tilde{\mathcal{X}}_i(t, t_\tau) $ in the time interval $[t_{\tau}, t_{\tau+1})$ is proved by introducing $S(\vect{e}_{i}) = 1 - e^{- V(\vect{e}_{i})}$. Note that $S$ satisfies: (i) $S(\vect{e}_{i}) \in [0,1)$ for all $\vect{e}_{i} \in \mathcal{D}$ with $\mathcal{D} = \bigtimes_{j=1}^{\eta_i}(-1,1)$, and (ii) $S(\vect{e}_{i}) \rightarrow 1$ as $\vect{e}_{i} \rightarrow \partial \mathcal{D}$. Therefore, invariance of the prescribed-performance domain can be established by showing that $S(\vect{e}_{i})$ remains bounded away from one.

    By replacing \eqref{eq: final Lyapunov function upper bound} and $V(\vect{e}_{i}) = - \ln (1-S(\vect{e}_{i}))$ in $\dot{S}(\vect{e}_{i}) = \dot{V}(\vect{e}_{i})(1-S(\vect{e}_{i}))$, $ \dot{S}(\vect{e}_{i}) \leq - \kappa (1-S(\vect{e}_{i})) \Bigl(-\frac{1}{\kappa}\beta(t) - \ln (1-S(\vect{e}_{i}))\Bigr)$ is obtained.
    Since $\kappa> 0$ and $1-S(\vect{e}_{i})>0$, verifying $\dot{S}(\vect{e}_{i}) \leq 0 $ reduces to checking when $-\frac{1}{\kappa}\beta(t) - \ln (1-S(\vect{e}_{i}))\geq 0$.
    Note that $-\frac{1}{\kappa}\beta(t) - \ln (1-S(\vect{e}_{i})) \geq -\frac{1}{\kappa}\bar \beta - \ln (1-S(\vect{e}_{i}))$, and $-\frac{1}{\kappa}\bar \beta - \ln (1-S(\vect{e}_{i})) \geq 0 $ for all $\vect{e}_{i} \in \Omega^c_{\vect{e}}$, where $\Omega^c_{\vect{e}}:= \mathcal{D}\setminus \Omega_{\vect{e}}$, with $\Omega_{\vect{e}}:=\bigl\{\vect{e}_{i} \in \mathcal{D}| S(\vect{e}_{i}) < 1- e^{-\frac{{\bar\beta}}{\kappa}} \bigr\}$. Moreover, $-\frac{1}{\kappa}\bar\beta - \ln (1-S(\vect{e}_{i})) = 0 $ holds for $\vect{e}_{i} \in  \partial\Omega_{\vect{e}}$. 
    Thus, $\dot{S}(\vect{e}_{i}) \leq 0$  for $\vect{e}_{i} \in \Omega^c_{\vect{e}}$.
    
    Since $|\xi^{j}_{i}(t_\tau)| < \underline{\rho}_i(t_\tau, t_\tau)$ is guaranteed by design if $T_\tau$ satisfies \eqref{eq: dwell time condition} and if $|\xi^{j}_{i}(t)| < \underline{\rho}_i(t, t_{\tau-1})$ for $t \in [t_{\tau-1}, t_\tau)$, $e^{j}_{i}(t_\tau, t_\tau) \in (-1,1)$ holds for all $j \in \Outhop{i}{k}$ by induction, and $S(\vect{e}_{i}(t_\tau, t_\tau)) < 1$ holds under the stated assumption.
    
    Then, since $\dot S(\vect{e}_i) \leq 0$ for all $\vect{e}_i \in \Omega^c_{\vect{e}}$, it
    follows that, regardless of whether $\vect{e}_i(t_\tau,t_\tau) \in \Omega_{\vect{e}}$ or
    $\vect{e}_i(t_\tau,t_\tau) \in \Omega^c_{\vect{e}}$,
    \begin{equation*}
        S(\vect{e}_i(t)) \leq \bar S := \max\bigl\{S(\vect{e}_i(t_\tau,t_\tau)),\,
        1-e^{-\bar\beta/\kappa}\bigr\} < 1
    \end{equation*}
    on any interval $[t_\tau, t^*) \subseteq [t_\tau, t_{\tau+1})$ on which the solution exists. As a result, $V(\vect{e}_i)$, and thus $\vect{\epsilon}_i$, remain bounded, so that $\vect{e}_i$ stays in a compact subset of $\mathcal{D}$. Therefore, every solution can be continued, and the local existence of a Carath\'eodory solution of   \eqref{eq: state estimation erro dynamics in time varying case} extends to $[t_\tau,t_{\tau+1})$, with $\tilde{\mathcal{X}}_i(t,t_\tau)$ forward invariant on this interval. Hence, the inverse transformation guarantees that $|\xi_i^j(t)|<\underline{\rho}_i(t,t_\tau)$ for all $t \in [t_\tau,t_{\tau+1})$ and $j\in\Outhop{i}{k}$. Under the funnel design conditions in Proposition~\ref{Proposition: underline rho lower than rho} and \eqref{eq: design of rho_i based on the minimum delta}, this further implies $|\error{x}{j}{i}(t)|<\delta_i^j(t)$ for all $j\in\Outhop{i}{k}$.
    
    Since the prescribed performance condition is preserved on the arbitrary interval $[t_\tau,t_{\tau+1})$, and the minimum dwell time ensures $t_\tau \to \infty$ as $\tau \to \infty$, the induction argument establishes its preservation over all switching intervals. Therefore, the prescribed performance requirements hold for all $t\geq0$.
\end{proof}

\section{\texorpdfstring{$k$}{k}-hop PPSO-Based Controller}
\label{Sec: Closed-loop}
In the previous sections, a \mbox{$k$-hop} PPSO has been developed to reconstruct the states of agents that are not directly available through the communication topology. This section establishes sufficient conditions under which the convergence guarantees of a controller designed with full state information are preserved when implemented using the estimated states. In particular, we characterize the effect of estimation errors on the satisfaction of the desired objectives and derive conditions ensuring convergence to a neighborhood of the nominal objective.

\subsection{Task-communication graph compatibility}
To characterize the dependency of the control objective on the agent states, we introduce the time-invariant task graph $\mathcal{G}_{\mathcal{T}}:= (\mathcal{V}, \mathcal{E}_{\mathcal{T}})$. An edge
$(j,i) \in \mathcal{E}_{\mathcal{T}}$ indicates that the task associated with agent $i$ depends on the state of agent $j$. The desired system behavior is represented by the target set $\mathcal{A}$, while $\mathcal{G}_{\mathcal{T}}$ specifies the state information required to achieve this objective. Accordingly, define $\TNeigh{i} := \{ j \in \mathcal{V}: (j,i) \in \mathcal{E}_{\mathcal{T}} \lor j = i\}$ as the set of closed in-neighbors of agent $i$ in the task graph. 

Since the communication topology is time-varying, the following assumption is imposed to ensure that every state required by the task graph is either directly available or reconstructible using the \mbox{$k$-hop} PPSO.
\begin{assumption}
    \label{Assumption: each communication and task graph compatibility}
    Consider the task graph $\mathcal{G}_{\mathcal{T}}$ and the set of communication graphs $\mathbb{G}$ defined in Section~\ref{Section: communication graph}. For every $\mathcal{G}_{l}\in\mathbb{G}$ and every $(j,i)\in\mathcal{E}_{\mathcal{T}}\setminus\mathcal{E}_{l}$, there exists a directed path from $j$ to $i$ in $\mathcal{G}_{l}$ with length at most $k$.
\end{assumption}

Assumption~\ref{Assumption: each communication and task graph compatibility} requires that every task dependency that is not directly available through the communication topology remains reachable within $k$ communication hops. If every task dependency is reachable in every communication topology, the assumption can always be satisfied by selecting $k$ sufficiently large. In particular, one may select
\begin{equation}
    k \geq \max_{\mathcal{G}_{l}\in\mathbb{G}}
\max_{(j,i)\in\mathcal{E}_{\mathcal{T}}\setminus\mathcal{E}_{l}}
D^{ji}_{l}
\end{equation}
where $D^{ji}_{l}$ denotes the length of the shortest directed path from $j$ to $i$ in $\mathcal{G}_{l}$. Thus, the assumption fundamentally requires that, in every communication graph, agents involved in a collaborative task that are not directly connected by a communication edge are reachable from one another through a directed path.
\subsection{Multi-agent system dynamics and controller}
\label{section: multi-agent system dynamics and controller}
Consider the vectorized form of the dynamics in \eqref{eq: agent's dynamic}, i.e., $\dot{\vect{x}}(t) = \vect{f}(\vect{x}(t)) + \vect{g}(\vect{u}(t)) + \vect{w}(\vect{x},t)$, where $\vect{x}(t) := [x_i]_{i\in \mathcal{V}}$, $\vect{f}(\vect{x}(t)) := [f_i(x_i)]_{i\in \mathcal{V}}$, $\vect{g}(\vect{u}(t)) := [g_i(u_i)]_{i\in \mathcal{V}}$, $\vect{w}(\vect{x},t) := [w_i(\vect{x},t)]_{i\in \mathcal{V}}$, and $\vect{u}(t) := [u_i(t)]_{i\in \mathcal{V}}$ denotes a nonlinear control input of the form
\begin{equation}
    \label{eq: ideal closed loop controller}
    \vect{u}^{\text{nom}} = \vect{\psi}(\vect{x}) := \left[\psi_i(\vect{x}_\mathcal{T}^i)\right]_{i\in \mathcal{V}},
\end{equation}
where $\vect{x}_\mathcal{T}^i := [x_j(t)]_{j \in \TNeigh{i}}$.
Since $\mathcal{G}_{\mathcal{T}}$ is time-invariant, the functional form of the controller is independent of the active communication graph $\mathcal{G}_{\sigma(t)}$. The active communication topology determines how the components of $\vect{x}_{\mathcal{T}}^i$ are obtained: each required component is either directly available through a communication edge in $\mathcal{G}_{\sigma(t)}$, or reconstructed by the observer from the available $k$-hop information.

Accordingly, the controller in \eqref{eq: ideal closed loop controller} is implemented using the available states and the corresponding observer estimates as
\begin{equation}
    \label{eq: real closed loop controller}
    \vect{u} = \left[\psi_i(\underline{\vect{x}}_\mathcal{T}^i )\right]_{i\in \mathcal{V}},
\end{equation}
where, for each $i \in \mathcal{V}$, $\underline{\vect{x}}_\mathcal{T}^i := [\underline{x}^i_j(t)]_{j \in \TNeigh{i}}$ with
\begin{equation*}
    \underline{x}^i_j(t) := \begin{cases}
        x_j(t),& \quad \text{ if } j \in \TNeigh{i} \cap \inN{i}(t) \text{ or } j = i, \\
        \hat{x}^i_j(t),& \quad \text{ otherwise.}
    \end{cases}
\end{equation*}
By writing $\underline{\vect{x}}_\mathcal{T}^i = \vect{x}_\mathcal{T}^i + \Tilde{\underline{\vect{x}}}_\mathcal{T}^i$, with $\Tilde{\underline{\vect{x}}}_\mathcal{T}^i := [\Tilde{\underline{x}}^i_j(t)]_{j \in \TNeigh{i}}$ and
\begin{equation*}
    \Tilde{\underline{x}}^i_j(t) := \begin{cases}
        0,& \quad \text{ if } j \in \TNeigh{i} \cap \inN{i}(t) \text{ or } j = i, \\
        \Tilde{x}^i_j(t),& \quad \text{ otherwise,}
    \end{cases}
\end{equation*}
the controller in \eqref{eq: real closed loop controller} can be expressed as $\vect{u} = \left[\psi_i(\vect{x}_\mathcal{T}^i + \Tilde{\underline{\vect{x}}}_\mathcal{T}^i)\right]_{i\in \mathcal{V}}$. 
Hence, communication topology switches and observer dynamics affect the closed-loop system only through the perturbations $\Tilde{\underline{\vect{x}}}_\mathcal{T}^i$, which can be regarded as external inputs induced by the state estimation errors. This motivates the use of the following set-Input to State Stability property to characterize the robustness of the closed-loop system with respect to such perturbations.

\begin{definition}[\!\!\cite{Sontag2008,Edwards-setISS}]
    \label{Definition set-ISS definition fomal}
    Let $\mathcal{A} \subset \mathbb{R}^n$ be nonempty and closed. A system $\dot{x}=f(x,u,t)$, with $f:\mathbb{R}^n \times \mathbb{R}^m \times \mathbb{R}_{\geq 0} \rightarrow \mathbb{R}^n$ continuous, is uniformly set-Input-to-State Stable (set-ISS) with respect to $\mathcal{A}$ if there exist $\beta \in \mathcal{KL}$ and $\gamma \in \mathcal{K}$ such
    that, for all $t_0 \geq 0$, all $x(t_0) \in \mathbb{R}^n$, and all measurable, essentially bounded inputs $u$, every solution $x(t)$ exists for all $t \geq t_0$ and satisfies
    \begin{equation}
        \label{eq: set-iss inequlity}
        \norm{x(t)}_{\mathcal{A}} \leq \beta\bigl(\norm{x(t_0)}_{\mathcal{A}},\, t-t_0\bigr)
        + \gamma\Bigl(\operatorname*{ess\,sup}_{t_0 \leq s \leq t}\norm{u(s)}\Bigr),
    \end{equation}
    for all $ t \geq t_0$, where $\norm{x}_{\mathcal{A}} := \inf_{a \in \mathcal{A}}\norm{x-a}$.
\end{definition}
\begin{remark}
    When $\mathcal{A}$ is compact, the requirement that every solution exists for all $t \geq t_0$ is redundant. Indeed, since the input is essentially bounded, the set-ISS estimate bounds $\norm{x(t)}_{\mathcal{A}}$ on the interval of existence, which, for compact $\mathcal{A}$, confines $x(t)$ to a compact set and rules out finite escape \cite{Edwards-setISS}.
\end{remark}

\begin{assumption}
    \label{Assumption on regular controller}
    The ideal controller $\vect{u}=\vect{\psi}(\vect{x})$ in \eqref{eq: ideal closed loop controller} renders the nominal closed-loop system uniformly globally asymptotically stable with respect to $\mathcal{A}$ for all disturbances $\vect{w}(\vect{x},t)$ satisfying Assumption~\ref{Assumption on existence of a solution}.
\end{assumption}

If $\vect{w}(x,t)$ does not vanish on the nominal target set, $\mathcal{A}$ is understood as a neighborhood of the nominal target set that is rendered globally asymptotically stable despite $w$, whose size depends on $\sup_{x,t}\norm{\vect{w}(\vect{x},t)}$.
Thus, $\vect{w}$ can be considered part of the nominal dynamics, and the system under \eqref{eq: real closed loop controller} can be studied by analyzing $\dot{\vect{x}}
= \Phi(\vect{x},\tilde{\underline{\vect{x}}}_{\mathcal T},t):= \vect{f}(\vect{x}) + \vect{g} \left(\left[\psi_i \left(\vect{x}_{\mathcal T}^i + \tilde{\underline{\vect{x}}}_{\mathcal T}^i \right) \right]_{i\in\mathcal V}\right) + \vect{w}(\vect{x},t)$, where $\tilde{\underline{\vect{x}}}_{\mathcal T}:=[\tilde{\underline{\vect{x}}}^i_{\mathcal T}]_{i \in \mathcal{V}}$ is treated as an input disturbance affecting the nominal closed-loop dynamics $\dot{\vect{x}} = \Phi(\vect{x},0,t)$.

For each $i\in \mathcal{V}$, define the stacked vector containing the upper bounds on the state estimation errors associated with the agents $j \in \TNeigh{i}$ at time $t$ as $\delta_{\tilde{\underline{\vect{x}}}^i_{\mathcal T}}(t):= \left[\underline{\delta}^i_{j}\right]_{j\in \TNeigh{i}}$, where
\begin{equation*}
    \underline{\delta}^i_{j} := \begin{cases}
        0,& \quad \text{ if } j \in \TNeigh{i} \cap \inN{i}(t) \text{ or } j = i, \\
        \delta^i_j(t),& \quad \text{ otherwise.}
    \end{cases}
\end{equation*}
Furthermore, let $\delta_{\tilde{\underline{\vect{x}}}_{\mathcal T}}(t) := [\delta_{\tilde{\underline{\vect{x}}}^i_{\mathcal T}}(t)]_{i \in \mathcal{V}}$.

\begin{assumption}
    \label{Assumption on the desired steady state estimation error}
    For a prescribed constant $\delta^{U}_{\tilde{\underline{\vect{x}}}_\mathcal{T}} > 0$, there exists a finite time $t_x \geq 0$ such that $\norm{\delta_{\tilde{\underline{\vect{x}}}_{\mathcal T}}(t)} \leq \delta^{U}_{\tilde{\underline{\vect{x}}}_\mathcal{T}}$ holds for all $t \geq t_x$.
\end{assumption}

Assumption~\ref{Assumption on the desired steady state estimation error} imposes a requirement on the evolution of the state estimation errors in $\Tilde{\underline{\vect{x}}}_\mathcal{T}(t)$. However, since the $k$-hop PPSO guarantees satisfaction of \eqref{eq: state estimation error inequality}, and $\delta_i^j(t)$ are design parameters, this requirement can be accommodated through an appropriate choice of the \mbox{steady-state} values of $\delta_i^j(t)$, provided that the corresponding choice of $\rho_i^j(t)$ satisfies the observer feasibility condition in \eqref{eq: asymptotic feasibility condition}. Consequently, the observer feasibility condition may impose a lower bound on the achievable steady-state values of $\delta_i^j(t)$, and hence on the admissible values of $\delta^{U}_{\tilde{\underline{\vect{x}}}_\mathcal{T}}$.

\begin{theorem}
    \label{Theorem on closed loop stability}
     Consider the heterogeneous MAS \eqref{eq: agent's dynamic} communicating over the directed switching graph $\mathcal{G}_{\sigma(t)} \in \mathbb{G}$, equipped with the decentralized state observers \eqref{Eq: State observer dynamics in components} and with a control objective described by a time-invariant graph $\mathcal{G}_{\mathcal T}$. Suppose each agent implements the control law in \eqref{eq: real closed loop controller}. Then, under Assumptions~\ref{Assumption on existence of a solution}-\ref{Assumption on the desired steady state estimation error}, if the dwell times satisfy \eqref{eq: dwell time condition}, the prescribed performances are designed as in Section~\ref{Sec: Decentralized PPSO}, and $\Phi(\vect{x},\tilde{\underline{\vect{x}}}_{\mathcal T},t)$ is uniformly \mbox{set-ISS} with respect to $\mathcal{A}$, the MAS converges to $\mathcal{A}_e := \{\vect{x}: \norm{\vect{x}}_{\mathcal{A}} \leq \gamma (\delta^{U}_{\tilde{\underline{\vect{x}}}_\mathcal{T}})\}$.
\end{theorem}

\begin{figure*}[t!]
    \centering
    \includegraphics[width =\linewidth]{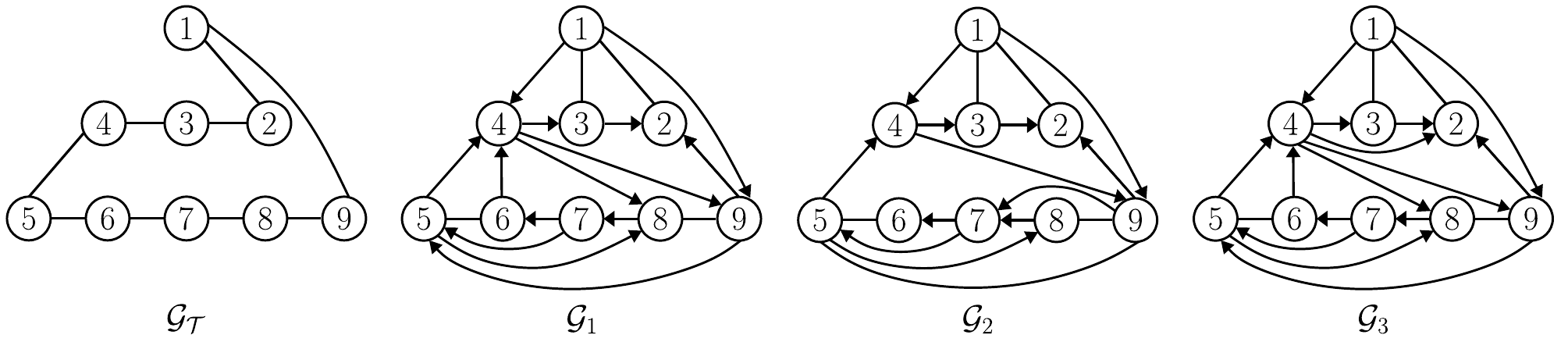}
    \caption{Fixed task graph $\mathcal{G}_{\mathcal{T}}$ and admissible communication graph set $\mathbb{G} = \{\mathcal{G}_1, \mathcal{G}_2, \mathcal{G}_3\}$. Directed edges are represented by arrows from the information-sending agent to the receiving agent, whereas undirected edges are represented by lines.} 
    \label{fig: Graph for the simulation and case study}
\end{figure*}

\begin{proof}
    Since $\Phi$ is continuous by Definition~\ref{Definition set-ISS definition fomal}, and~\eqref{Eq: State observer dynamics in components} is continuous on the prescribed-performance domain, the closed-loop system admits, on each switching interval, a local Carath\'eodory solution~\cite{Cortes_DiscontinuousDS}. Moreover, every such solution exists for all $t\geq0$. Indeed, on any interval on which it exists, Theorem~\ref{Theorem: main theorem on state estimation convergence in time varying case} guarantees that $\tilde{\underline{\vect{x}}}_{\mathcal{T}}$ remains bounded by the prescribed performance functions. Extending $\tilde{\underline{\vect{x}}}_{\mathcal{T}}$ by zero beyond this interval yields a bounded input defined for all $t\geq0$, so that the forward completeness included in Definition~\ref{Definition set-ISS definition fomal} rules out finite escape times of $\vect{x}$ (see also Remark~\ref{remark: agent trajectories exist}). Hence, by a standard extension argument, no solution can terminate in finite time.
    Under Assumptions~\ref{Assumption on existence of a solution}-\ref{Assumption on 2-hop communication}, dwell times satisfying \eqref{eq: dwell time condition}, and a proper design of the prescribed performance functions, Theorem~\ref{Theorem: main theorem on state estimation convergence in time varying case}, applied to each agent $j$ with $i \in \Outhop{j}{k}(t)$, guarantees $|\tilde{x}_{j}^{i}(t)| < \delta_{j}^{i}(t)$ for all $i \in \mathcal{V}$, $j \in \Inhop{i}{k}(t)$, and $t \in \mathbb{R}_{\geq 0}$. Moreover, Assumption~\ref{Assumption: each communication and task graph compatibility} guarantees $\TNeigh{i} \setminus (\inN{i}(t) \cup \{i\}) \subseteq \Inhop{i}{k}(t)$ for every active communication topology. Hence, the above prescribed performance bound applies to every task-relevant state that is not directly available to agent $i$. Thus, under Assumption~\ref{Assumption on the desired steady state estimation error}, $\delta_{j}^{i}(t)$ can be designed such that there exists $t_x\geq 0$ satisfying $\norm{\delta_{\tilde{\underline{\vect{x}}}_{\mathcal T}}(t)} \leq \delta^{U}_{\tilde{\underline{\vect{x}}}_\mathcal{T}}$ for all $t \geq t_x$. 
     
     For $t \geq t_x$, the closed-loop state evolves from $\vect{x}(t_x)$ according to $\dot{\vect{x}} = \Phi(\vect{x},\tilde{\underline{\vect{x}}}_{\mathcal T},t)$. Moreover, since $|\tilde{x}_{j}^{i}(t)| < \delta_{j}^{i}(t)$ for all $i \in \mathcal{V}$ and $j \in \Inhop{i}{k}(t)$, it holds that $\norm{\Tilde{\underline{\vect{x}}}_\mathcal{T}(t)} \leq \norm{\delta_{\tilde{\underline{\vect{x}}}_{\mathcal T}}(t)} \leq \delta^{U}_{\tilde{\underline{\vect{x}}}_\mathcal{T}}$ for $t\geq t_x$. Thus, by the uniform set-ISS property of $\Phi$, $\norm{\vect{x}(t)}_{\mathcal{A}} \leq \beta(\norm{\vect{x}(t_x)}_{\mathcal{A}},t-t_{x}) + \gamma(\delta^{U}_{\tilde{\underline{\vect{x}}}_\mathcal{T}})$. Since $\beta$ is a $\mathcal{KL}$ function, $\beta(\norm{\vect{x}(t_{x})}_{\mathcal{A}},t-t_{x}) \to 0$ as $t \rightarrow \infty$. Hence, $\limsup_{t\to\infty}\norm{\vect{x}(t)}_{\mathcal{A}} \leq \gamma(\delta^{U}_{\tilde{\underline{\vect{x}}}_\mathcal{T}})$, i.e., the MAS converges to $\mathcal{A}_e = \{\vect{x}: \norm{\vect{x}}_{\mathcal{A}} \leq \gamma(\delta^{U}_{\tilde{\underline{\vect{x}}}_\mathcal{T}})\}$.
      
\end{proof}

\begin{remark}
    \label{remark: restricted set-ISS}
    Theorem~\ref{Theorem on closed loop stability} still holds if \eqref{eq: set-iss inequlity} is only satisfied on a set $\mathcal{X}$ and for state estimation errors bounded by $\bar\delta \geq \delta^U_{\tilde{\underline{\vect{x}}}_\mathcal{T}}$, provided that $\vect{x}(t_x) \in \mathcal{X}$ and $\mathcal{X}$ is forward invariant for
    such errors. Indeed, a local Carath\'eodory solution exists by continuity of $\Phi$ and of the
    observer dynamics, as in the proof of Theorem~\ref{Theorem on closed loop stability}, and Remark~\ref{remark: agent trajectories exist} rules out finite
    escape, so solutions exist for all $t \geq 0$; moreover, the estimate is
    only used for $t \geq t_x$.
\end{remark}

\section{Numerical Simulation}
\label{Sec: Case Study}
Consider a MAS of $N= 9$ agents communicating over the directed switching graph $\mathcal{G}_{\sigma(t)} \in \mathbb{G}$, where $\sigma : [0, \infty) \rightarrow \{1,2,3\}$ is a piecewise switching signal and $\mathbb{G}= \{\mathcal{G}_1,\mathcal{G}_2, \mathcal{G}_3\}$, with $\mathcal{G}_1 $, $\mathcal{G}_2$ and $\mathcal{G}_3$ depicted in Fig.~\ref{fig: Graph for the simulation and case study}.

Each agent is modeled by the three-dimensional translational point-mass dynamics $\dot{x}_i =  f(x_i) + g_i u_i$, where $x_i = [p_i^\top,v_i^\top]^\top \in \mathbb{R}^6$ is the stacked position-velocity vector, $p_i,v_i \in \mathbb{R}^3$ denote position and velocity in the world frame, and  $u_i \in \mathbb{R}^3$ is the virtual acceleration input. The drift and input maps are given by
\begin{equation*}
    f(x_i) = \begin{bmatrix}
                v_i \\ - g_0\hat{e}_3 -\kappa_d \sat(v_i, v_{\sat})
            \end{bmatrix} \text{ and } g_i \equiv g = \begin{bmatrix}
                0_3 \\ I_3
            \end{bmatrix},
\end{equation*}
where $g_0 = 9.81  \mathrm{m/s^2}$, $\hat{e}_3 := [0, 0, 1]^\top$, $\kappa_d = 0.5$, $v_\sat = 15\mathrm{m/s}$, and $\sat(\cdot, \cdot): \mathbb{R}^{3} \times \mathbb{R}_{> 0} \rightarrow  \mathbb{R}^3$ is a saturation function modeling a bounded drag acceleration, defined as
\begin{equation*}
    \sat(y, M):= \begin{cases}
        M \tanh\left(\frac{\norm{y}}{M}\right) \frac{y}{\norm{y}}, \quad &y\neq 0, \\
        0, \quad &y =  0.
    \end{cases}
\end{equation*}

The MAS is required to achieve a desired formation specified by the fixed undirected task graph $\mathcal{G}_{\mathcal{T}}:= (\mathcal{V}, \mathcal{E}_{\mathcal{T}})$, depicted in Fig.~\ref{fig: Graph for the simulation and case study}, where each edge $(j,i) \in \mathcal{E}_{\mathcal{T}}$ represents a relative position constraint associated with a desired constant offset $d_{ij} = - d_{ji}\in \mathbb{R}^3$. Specifically, the agents are required to form a regular polygon inscribed in a circle of radius $150 \mathrm{m}$ at a desired altitude $h^* = 100 \mathrm{m}$. Although $\mathcal{G}_{\mathcal{T}}$ is undirected, the state information required to evaluate the corresponding relative-position constraints is not necessarily available to both agents through the active communication topology. Such information must be therefore reconstructed using the proposed \mbox{$k$-hop} observer.

Let $z_{ij} := p_i - p_j - d_{ij}$, $\hat{z}_{ij} := p_i - \hat{p}^i_j - d_{ij}$, $w_{ij} := v_i - v_j$, $\hat{w}_{ij} := v_i - \hat v^i_j$, $\tilde{h} := [p_1]_3 - h^*$ and $\bar w := [v_1]_3$. To achieve the desired formation despite the mismatch between the communication and task graphs, each agent implements
\begin{equation}
    \begin{split}
        \label{eq: real input in the simulation}
        u_i(t) = g_0 \hat{e}_3 - \kappa_p (u_i^p + \hat{u}_i^p) - \kappa_v (u_i^v+ \hat{u}_i^v) - u^h_i \hat{e}_3,
    \end{split}
\end{equation}
where
\begin{equation*}
    \begin{split}
        \hspace{-0.1cm}u_i^p &= \hspace{-0.2cm} \sum_{j \in \TNeigh{i} \cap \inN{i}(t)} \hspace{-0.2cm} \sat(z_{ij}, M_p),  \       \hat{u}_i^p =  \hspace{-0.2cm} \sum_{j \in \TNeigh{i} \cap \Inhop{i}{k}(t)} \hspace{-0.2cm} \sat(\hat z_{ij}, M_p), \\
        \hspace{-0.1cm}u_i^v &= \hspace{-0.2cm} \sum_{j \in \TNeigh{i} \cap \inN{i}(t)} \hspace{-0.2cm} \sat(w_{ij}, M_v), \        \hat{u}_i^v = \hspace{-0.2cm} \sum_{j \in \TNeigh{i} \cap \Inhop{i}{k}(t)} \hspace{-0.2cm} \sat(\hat w_{ij}, M_v)
    \end{split}
\end{equation*}
and $u^h _i $ is an altitude-pinning term given by
\begin{equation}
    u^h_i = \begin{cases}
        \kappa_p \sigma_{M_p}(\tilde{h}) + \kappa_v \sigma_{M_v}(\bar w), & \quad \text{if } i = 1, \\
        0, & \quad \text{otherwise,}        
    \end{cases}
\end{equation}
where $\sigma_M(s):= M \tanh\left(\frac{s}{M}\right)$, $s\in \mathbb{R}$.
The terms $u_i^p$ and $u_i^v$ use states directly available through the active communication graph, whereas $\hat u_i^p$ and $\hat u_i^v$ use the corresponding \mbox{$k$-hop} state estimates generated by the proposed observer. Thus, \mbox{task-relevant} states that are not directly available to agent $i$ can still be incorporated into its formation controller. The \mbox{altitude-pinning} term $u^h_i$ is applied only to agent $1$ with the ultimate goal of regulating the formation altitude to $h^*$, thereby fixing the translational degree of freedom along the $z$ axis. The gains $\kappa_p , \kappa_v, M_p$ and $ M_v \in \mathbb{R}_{> 0}$ are design parameters. 

Note that in order to validate the observer-based controller for the nonlinear case, we do not cancel the drift term $f(x_i)$ through the control input in \eqref{eq: real input in the simulation}.

Since Assumptions~\ref{Assumption on existence of a solution}-\ref{Assumption on 2-hop communication} hold, Theorem~\ref{Theorem: main theorem on state estimation convergence in time varying case} applies and guarantees convergence of the state estimation error for dwell times satisfying
\eqref{eq: dwell time condition}, provided that condition~\eqref{eq: asymptotic feasibility condition} is satisfied. 
With observer convergence established, the remaining conditions of Theorem~\ref{Theorem on closed loop stability}, verified in the sense of Remark~\ref{remark: restricted set-ISS}, are the global asymptotic stability of the MAS to $\mathcal A:=\{x\in\mathbb{R}^{54}\ :\ z_{ij}=0 \   \forall(j,i)\in\mathcal{E}_{\mathcal T}, \tilde{h}=0, v=0\}$ under the nominal input $u_i^{\text{nom}} = g_0 \hat{e}_3 - \kappa_p \sum_{j \in \TNeigh{i}} \sat(z_{ij}, M_p)  - \kappa_v \sum_{j \in \TNeigh{i}} \sat(w_{ij}, M_v) - u^h_i$, and the set-ISS property of the resulting closed-loop MAS with respect to the estimation error when the true input is applied. In particular, the latter is established by considering the \mbox{observer-based} input $u_i^{\text{est}} = g_0 \hat{e}_3 - \kappa_p \sum_{j \in \TNeigh{i}} \sat(z_{ij} - \Tilde{\underline{p}}^i_j(t), M_p)  - \kappa_v \sum_{j \in \TNeigh{i}} \sat(w_{ij} -\Tilde{\underline{v}}^i_j(t), M_v) - u^h_i$, which is obtained by replacing each unavailable state $x_j$ with its estimate $x_j+ \Tilde{\underline{x}}^i_j(t)$, where $\Tilde{\underline{x}}^i_j(t)$ is defined as in \ref{section: multi-agent system dynamics and controller}.

To establish convergence of the nominal MAS, define $q_i := p_i-d_i$ for all $i\in\mathcal{V}$, where $d_i\in\mathbb{R}^3$ satisfy $d_{ij}=d_i-d_j$ for each $(j,i)\in\mathcal{E}_{\mathcal{T}}$ and $[d_i]_3=h^*$, so that $z_{ij}=q_i-q_j$ and $\tilde{h}=[q_1]_3$. The convergence of the MAS to $\mathcal{A}$ follows from standard Lyapunov arguments using $V(\vect{q}, \vect{v})= \frac{1}{2}\kappa_p \sum_{(j,i)\in\mathcal{E}_{\mathcal{T}}} \Psi(q_i - q_j, M_p) + \kappa_p \Psi_1(\tilde{h}, M_p) + \frac{1}{2}\sum_{i=1}^{N} \norm{v_i}^2$, where $\Psi(x,M) := M^2 \ln\cosh\!\big(\norm{x}/M\big)$ and $\Psi_1(s, M) = M^2 \ln \cosh(s/M)$. Similar analysis for a single integrator is found in \cite{zhang2011consensus}.
For robustness, let $q_{\perp,i} := q_i - \frac{1}{N}\sum_{j}\Pi_h q_j$, with $\Pi_h = I_3 - \hat{e}_3\hat{e}_3^\top$, and $V_1 := V + \epsilon\sum_{i} q_{\perp,i}^\top v_i$. Since $\mathrm{sat}(\cdot,M)$ is Lipschitz, the estimation errors act as an additive perturbation bounded by $\norm{\tilde{\underline{\vect{x}}}_{\mathcal T}}$. For any $c > 0$ and sufficiently small $\epsilon > 0$, $V_1$ is equivalent to $V$ on $\{V \leq c\}$ and yields \eqref{eq: set-iss inequlity} on a sublevel set $\mathcal{X}$ of $V_1$ contained in $\{V \leq c\}$, for $\norm{\tilde{\underline{\vect{x}}}_{\mathcal T}} \leq \bar\delta$, where, for a given $\bar{\delta}$, $\mathcal{X}$ can be enlarged by increasing $M_p$ and $M_v$. Choosing $c$ such that $\vect{x}(t_x) \in \mathcal{X}$, which is possible since trajectories are bounded on $[0,t_x]$ by Remark~\ref{remark: agent trajectories exist}, and $M_p$, $M_v$ such that $\delta^U_{\tilde{\underline{\vect{x}}}_\mathcal{T}} \leq \bar\delta$, Remark~\ref{remark: restricted set-ISS} applies.

\begin{figure}
    \centering
    \includegraphics[width =\linewidth]{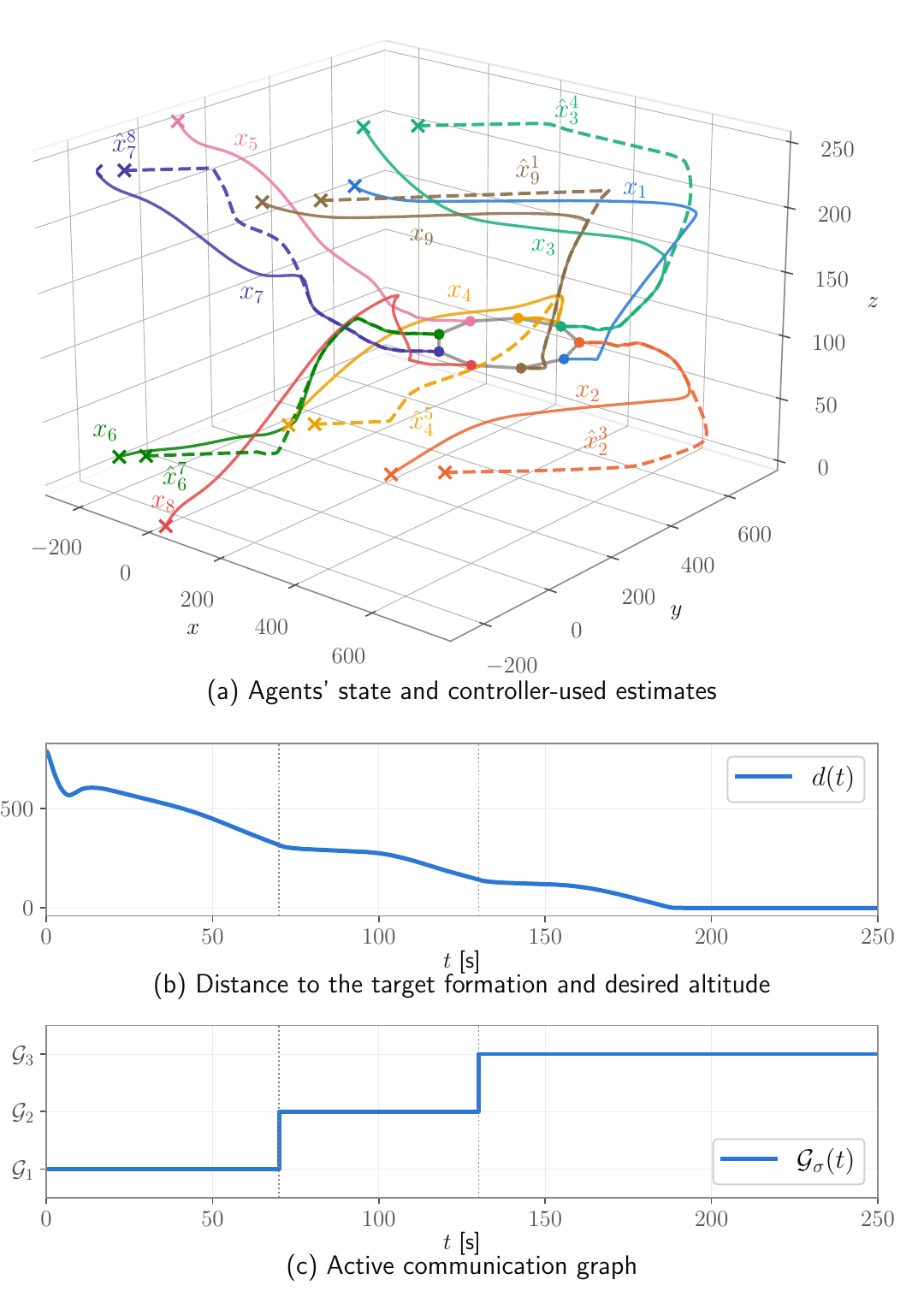}
    \caption{Formation control simulation with nine agents. (a) Agent trajectories and the state estimates used in the controllers. (b) Aggregate formation and altitude error evolution over time. (c) Switching communication topologies.}
    \label{fig: MAS trajectory}
\end{figure}

Fig.~\ref{fig: MAS trajectory} and Fig.~\ref{fig: State estimation error and disagreemenr} present simulation results when the MAS is controlled with  \eqref{eq: real input in the simulation}. To enable every agent $i \in \mathcal{V}$ to estimate the states of agents $j \in \TNeigh{i} \setminus \inNeightime{i}$, the \mbox{$k$-hop} PPSO in \eqref{Eq: state estimation dynamics in vector form} is implemented with $k = 3$. For simplicity, the same prescribed bound is adopted for all agents and state components, i.e., $[\delta_i^j]_l =\delta_i(t) = \delta(t)$ for all $i \in \mathcal{V}$, $j \in \mathcal{O}_{i}^{\cup}$ and $l = 1,\dots, 6$. Accordingly, $\rho_i(t)$ and $\underline{\rho}_i(t, t_{\tau})$ are designed following the procedure in Section~\ref{Sec: Decentralized PPSO}. Fig.~\ref{fig: MAS trajectory}(a) depicts the trajectories of the closed-loop system and of the state estimate used in the control under the switching signal in Fig.~\ref{fig: MAS trajectory}(c). As shown in Fig.~\ref{fig: MAS trajectory}(b), the aggregate formation and altitude error $d(t):=  \sqrt{\sum_{(j,i)\in\mathcal{E}_{\mathcal T}}(\|z_{ij}(t)\|^2+\|w_{ij}(t)\|^2) + \|\tilde{h}(t)\|^2 + \|\bar{w}\|^2}$, i.e., the norm of the stacked error vector, converges to zero. Since every summand is nonnegative, this means that all edge errors and the altitude error vanish asymptotically, so the agents achieve the desired formation and altitude. This result demonstrates that the MAS objective is achieved even when the controller in \eqref{eq: real input in the simulation} is used in place of the ideal controller $u_i^{\text{nom}}$.
\begin{figure}
    \centering
    \includegraphics[width =\linewidth]{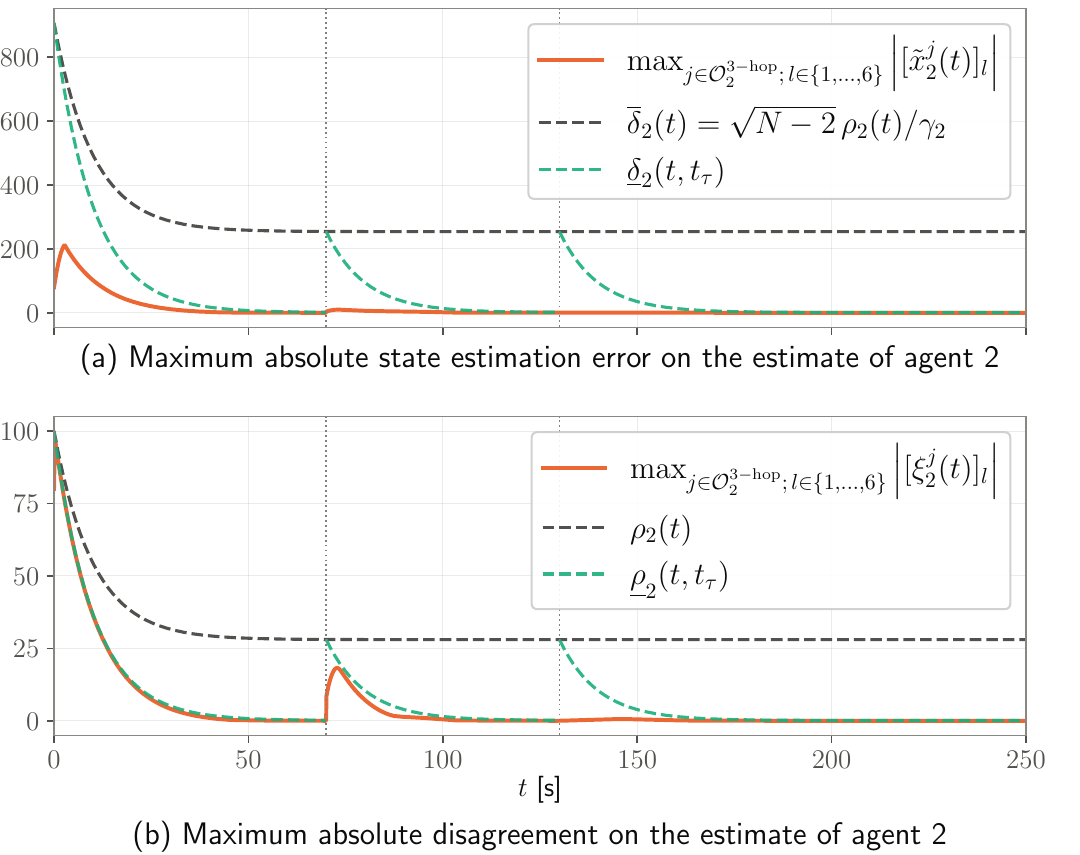}
    \caption{Prescribed performance bounds for estimates of agent 2. (a) Maximum absolute state estimation error for agent~$2$ maintained by the other agents. (b) Maximum absolute disagreement component in the estimates of agent~$2$. The dashed lines represent the prescribed performance functions of the \mbox{$k$-hop} PPSO. Specifically, the original prescribed performance functions $\overline{\delta}_2$ and $\rho_2$ are shown in black, while the auxiliary envelopes $\underline{\rho}_2(t,t_\tau)$ and $\underline{\delta}_2(t,t_\tau)$ are shown in green.}
    \label{fig: State estimation error and disagreemenr}
\end{figure}
For space limitations, Fig.~\ref{fig: State estimation error and disagreemenr} shows only the evolution of the maximum absolute state estimation error and the maximum disagreement among the estimates of the state of agent~$2$. 
As introduced earlier, for simplicity, $[\delta_2^j]_l  = \overline{\delta}_2(t) =  653.16e^{-0.1t} + 254.01$ for all $j \in \mathcal{O}_2^{\cup}$ and $l = 1,\dots, 6$. Then, since the finite set of admissible communication graphs $\mathbb{G}$ is known a priori, $\gamma_2(N) = \min_{\mathcal{G}_l \in \mathbb{G}} \{\sigma_{\min}(M_2)\} \approx 0.29$ and $\alpha_2(N) = \max_{\mathcal{G}_{l} \in \mathbb{G}} \{\sigma_{\max}(M_2)\} \approx 4.78$, and the resulting funnel for the disagreement components is $\rho_2(t) =  72e^{-0.1 t} + 28$. As discussed in Remark~\ref{Remark: phi can be made arbitrarily small}, an appropriate initialization procedure and design of $\underline{\rho}_i(t, t_{\tau})$ allow $\phi_i^{\max}$ to be selected arbitrarily small. Nevertheless, to clearly illustrate the effect of the prescribed-performance bounds, a relatively large value of $\phi_i^{\max}=4$ is selected for the simulations.
Given the $\phi_i^{\max}$ and the switching instants, the tighter prescribed performance functions are designed as $\underline{\rho}_2(t, t_{0}) = 99.898 e^{-0.1 t} + 0.1024$, $\underline{\rho}_2(t, t_{1}) = 27.963 e^{-0.1(t-70)} + 0.1024$ and $\underline{\rho}_2(t, t_{2}) = 27.898 e^{-0.1(t-130)} + 0.1024$. Then, the resulting $\underline{\delta}_2(t, t_{0}) = 906.233 e^{-0.1t} + 0.929$, $\underline{\delta}_2(t, t_{1}) = 253.672 e^{-0.1(t-70)} + 0.929$ and $\underline{\delta}_2(t, t_{2}) = 253.078 e^{-0.1(t-130)} + 0.929$ are obtained by reverting the relationship between $\overline{\delta}_i(t)$ and $\rho_i(t)$.
Fig.~\ref{fig: State estimation error and disagreemenr} clearly shows that as long as the disagreement component satisfies $ |\xi^j_2| < \underline{\rho}^{j}_2(t, t_{\tau}) $, with properly designed $\underline{\rho}^{j}_2(t, t_{\tau})$, dwell times satisfying \eqref{eq: dwell time condition} guarantee that the state estimation error satisfies $|[\error{x}{j}{2}]_l| < [\delta_2^j]_l(t)$ for all $j \in \Outhop{2}{3}$, $l = 1, \dots, 6$, and every switching topology in $\mathbb{G}$.

\section{Conclusions}
\label{Sec: Conclusions}
This paper develops a decentralized $k$-hop PPSO for heterogeneous multi-agent systems communicating over directed switching topologies. The proposed observer enables the reconstruction of nonlocal agent states using only \mbox{$1$-hop} communication while guaranteeing predefined transient bounds on the state estimation errors. To account for topology switching, a novel prescribed performance funnel construction is introduced, and feasibility conditions and a lower bound on the \mbox{dwell-time} are derived to preserve the original prescribed performance specifications across topology switches. Furthermore, a \mbox{set-ISS} analysis establishes that, for control objectives achieved with exact nonlocal state information, replacing these states with their estimates preserves convergence to a neighborhood of the desired set, whose size can be made arbitrarily small through the observer performance bounds. 

Future work will focus on characterizing the topological properties of the communication graph that enable an edge-selection mechanism to dynamically reconfigure the communication topology so as to satisfy the prescribed-performance constraints. Furthermore, we will investigate conditions guaranteeing observer and control convergence under intermittent loss of communication paths required to estimate nonlocal states.

\appendix

\subsection{Proof of Lemma~\ref{lemma: neighbors}} 
\label{Appendix: proof lemma neihbors}
\begin{proof}
    Let's start by proving (i). Let $j\in\Outhop{i}{k}$. Since $j$ belongs to the $k$-hop neighborhood of $i$, $2\le D^{ij}_{\sigma(t)}\le k$ holds. If $D^{ij}_{\sigma(t)}=2$, then there exists a node $v_0$ such that $i \rightarrow v_0 \rightarrow j$. Hence, $v_0\in\outN{i}{}\cap\inN{j}{}$, which implies $\inN{j}{}\cap\outN{i}{}\neq\emptyset$. If $D^{ij}_{\sigma(t)}>2$, let $i=v_0\rightarrow v_1\rightarrow\cdots\rightarrow v_{D^{ij}_{\sigma(t)}}=j$ be a shortest directed path from $i$ to $j$. Then, the predecessor $v_{D^{ij}_{\sigma(t)}-1}$ of $j$ satisfies $v_{D^{ij}_{\sigma(t)}-1}\in\Outhop{i}{k}$. Since $v_{D^{ij}_{\sigma(t)}-1}\in\inN{j}{}$, it follows that $\inN{j}{}\cap\Outhop{i}{k}\neq\emptyset$. Therefore, $
    \inN{j}{}\cap\Outhop{i}{k}\neq\emptyset$ or $\inN{j}{}\cap\outN{i}{}\neq\emptyset$ holds, which proves (i).
    
    To prove (ii), suppose by contradiction that there exist $i\in\mathcal{V}$ and $\mathcal{G}_{\sigma(t)}^{i,(p)}$ such that $\inN{j}{}\cap\outN{i}{}=\emptyset$ for all $j\in\mathcal{V}_{\sigma(t)}^{i,(p)}$. For any $j \in \mathcal{V}_{\sigma(t)}^{i,(p)}$, since $\mathcal{V}_{\sigma(t)}^{i,(p)} \subseteq \Outhop{i}{k}$, there exists a directed path from $i$ to $j$ of length at most $k$. Let $i=v_0\rightarrow v_1\rightarrow\cdots\rightarrow v_r=j$, $2\le r\le k$, be a shortest such path. By minimality of the path, each vertex $v_l$, $2\leq l \leq r$, is at the shortest-path distance $D^{i,v_l}_ {\sigma(t)}\in [2,k]$ from $i$, and therefore $v_l \in \Outhop{i}{k}$. Consequently, the subpath $v_2\rightarrow v_3\rightarrow\cdots\rightarrow v_r$ is a path in $\mathcal{G}_{\sigma(t)}^{i}$, so $v_2$ and $j$ belong to the same weakly connected component. Hence, $v_2\in\mathcal{V}_{\sigma(t)}^{i,(p)}$. Moreover, $v_1\in\inN{v_2}{}\cap\outN{i}{}$, so $\inN{v_2}{}\cap\outN{i}{} \neq \emptyset$, which contradicts the assumption. Therefore, every weakly connected component of $\mathcal{G}_{\sigma(t)}^{i}$ contains at least one agent $j$ satisfying $\inN{j}{}\cap\outN{i}{}\neq\emptyset$.
\end{proof}

\subsection{Proof of Lemma~\ref{Lemma: -M_i is Hirwitz}}
\label{Appendix: proof lemma -M_i is Hirwitz}
\begin{proof}
    The proof is divided into three steps. In (i), we show that every eigenvalue of $M_i^{(p)}$ has a nonnegative real part and that the only possible eigenvalue on the imaginary axis is the origin. In (ii), we prove that zero cannot be an eigenvalue of $M_i^{(p)}$. Finally, in (iii), we use \eqref{eq: Spectrum of M_i equal to the one of the block diagonals} to extend the result from each block matrix $M_i^{(p)}$ to $M_i$. 

    Before proceeding with the proof, enumerate the agents in $\mathcal{V}_{\sigma(t)}^{i,(p)}$ and let $q_j$ denote the global index of the agent at position $j$.
    Furthermore, let $d_{q_j}^{i,\text{in}}$ be the cardinality of the set $\mathcal{N}^{i, \text{in}}_{q_j}(t):= \{q_r\in \mathcal{V}_{\sigma(t)}^{i,(p)}: (q_r, q_j) \in \mathcal{E}_{\sigma(t)}^{i,(p)}\}$. Note that, since no edge connects distinct weakly connected components, $\mathcal{N}^{i, \text{in}}_{q_j} = \{q_r \in \mathcal{V}_{\sigma(t)}^{i}: (q_r, q_j) \in \mathcal{E}_{\sigma(t)}^{i}\} = \inN{q_j} \cap \Outhop{i}{k}$.

    \textbf{Part (i).} Consider a block matrix $M_i^{(p)} = L_i^{(p)} + H_i^{(p)}$. The Gershgorin disc associated with row $j$ is centered at $[M_i^{(p)}]_{j,j} = d_{q_j}^{i,\text{in}} + h^i_{q_j}$ and has radius $R_j = \sum_{r \neq j} |[M_i^{(p)}]_{j,r}| = d_{q_j}^{i,\text{in}}$, where $h^i_{q_j}$ is the corresponding diagonal entry of the pinning matrix $H_i^{(p)}$. Hence, the leftmost point of the $j$-th Gershgorin disc has real part $(d_{q_j}^{i,\text{in}}+h^i_{q_j})-d_{q_j}^{i,\text{in}} = h^i_{q_j}\geq0$. Therefore, all Gershgorin discs are contained in the closed right half-plane, which implies $\Re(\lambda) \geq 0$ for all $\lambda \in \spec(M_i^{(p)})$. 
    Moreover, a Gershgorin disc can intersect the imaginary axis only when $h^i_{q_j}=0$, in which case the intersection is limited to the origin. Hence, the only eigenvalue that may have zero real part is $\lambda=0$.

    \textbf{Part (ii).} We now show that zero cannot belong to the spectrum of $M_i^{(p)}$. Equivalently, it suffices to show that $M_i^{(p)} v = 0$ admits only the trivial solution $v=0$. Suppose, by contradiction, that there exists a nonzero real vector $v$ satisfying $M_i^{(p)} v = 0$. 
    Since $v \neq 0$, $\max_{j}[v]_j>0$ or $\min_{j}[v]_j<0$ must hold. Since $M_i^{(p)}(-v)=0$ whenever $M_i^{(p)}v=0$, we may replace $v$ with $-v$, if necessary, and assume without loss of generality that $\mu := \max_{j} [v]_j > 0$. Define
    \begin{equation}
        S := \left\{ q_j : [v]_j = \mu \right\},
    \end{equation}
    which is nonempty by construction. 
    The $j$-th row of $M_i^{(p)}v=0$, for which $q_j\in S$, satisfies
    \begin{equation}
        \label{eq: p rows of the Mv matrix vector multiplication}
        (d_{q_j}^{i,\text{in}} + h^i_{q_j})\mu = \sum_{r \in \{r : q_r \in 
        \mathcal{N}^{i, \text{in}}_{q_j}\}} [v]_r.
    \end{equation}
    Since $[v]_r \leq \mu$ for every $r$, it follows that $(d_{q_j}^{i,\text{in}} + h^i_{q_j})\mu = \sum_{r \in \{r : q_r \in \mathcal{N}^{i, \text{in}}_{q_j}\}} [v]_r\leq d_{q_j}^{i,\text{in}} \mu$ and therefore that $h^i_{q_j} \mu \leq 0$ holds. Because $\mu > 0$ and $h^i_{q_j} \geq 0$, $h^i_{q_j} =  0$ for all $q_j \in S$.
    Moreover, substituting $h^i_{q_j}=0$ into \eqref{eq: p rows of the Mv matrix vector multiplication} yields $\sum_{r \in \{r : q_r \in \mathcal{N}^{i, \text{in}}_{q_j}\}} (\mu - [v]_r) = 0$. Since every term in the summation is nonnegative, each term must be zero. Hence, $[v]_r = \mu$ for all $r$ such that $q_r \in \mathcal{N}^{i, \text{in}}_{q_j}$, which implies $ \mathcal{N}^{i, \text{in}}_{q_j} \subseteq S$. Therefore, $S$ is closed under in-neighbors, and every node in $S$ is unpinned, i.e., $h^i_{q_j} = 0$ for all $q_j \in \mathcal{S}$.

    Now let $q_j\in S$ be arbitrary. Since $q_j\in\mathcal V_{\sigma(t)}^{i,(p)}\subseteq\Outhop{i}{k}$, there exists a shortest directed path $i=v_0 \rightarrow v_1\rightarrow\cdots\rightarrow v_l= q_j$ with $2\le l\le k$. By minimality of the path, and by definition of \mbox{$k$-hop} out-neighbors, $v_2\rightarrow\cdots\rightarrow v_l$ lies entirely within the induced subgraph $\mathcal G_{\sigma(t)}^{i,(p)}$. Since $q_j=v_l\in S$ and $S$ is closed under in-neighbors, backward propagation along this path gives $v_l \in S \implies v_{l-1} \in \mathcal{N}^{i, \text{in}}_{v_l} \subseteq S \implies \cdots \implies v_2  \in S$. Every node in $S$ is unpinned. However, since $v_1 \in \inN{v_2} \cap \outN{i}$, it follows that $\inN{v_2} \cap \outN{i} \neq \emptyset$. Hence, $h^i_{v_2}>0$ by definition of pinning matrix, which contradicts the original assumption. Therefore, $M_i^{(p)}v=0$ admits only the trivial solution, and $0\notin\spec(M_i^{(p)})$.

    \textbf{Part (iii).} By part (i), every eigenvalue of $M_i^{(p)}$ has nonnegative real part, and by Part (ii), $0 \notin \spec(M_i^{(p)})$. Hence, $\Re(\lambda) > 0$ for all $\lambda \in \spec(M_i^{(p)})$, and since $\spec(M_i) = \bigcup_{p = 1}^{l^i} \spec(M_i^{(p)})$ by \eqref{eq: Spectrum of M_i equal to the one of the block diagonals}, it follows that $\Re(\lambda) > 0$ for all $\lambda \in \spec(M_i)$.
\end{proof}

\subsection{Proof of Proposition~\ref{Proposition: proof of proposition on norm(M_i^-1)}}
\label{Appendix: proof of proposition on norm(M_i^-1)}
\begin{proof}
    If $N = 3$, then $\eta_i \leq N-2 = 1$, since $\Outhop{i}{k}$ excludes both $i$ and its out-neighbors. Hence, $M_i = H_i$ is a scalar satisfying $1 \leq M_i \leq N-1-\eta_i = 1$ by Lemma~\ref{lemma: neighbors}(ii), so that $\sigma_{\max}(M_i) = \sigma_{\min}(M_i) = 1$. In the following, let $N \geq 4$.
    
    (i) By construction, $M_i=L_i+H_i$, where $L_i$ is the in-degree Laplacian matrix of the subgraph induced by the $k$-hop neighborhood of agent $i$, and $H_i$ is the associated pinning matrix. Since the induced subgraph contains $\eta_i$ vertices, its maximum in- and out-degrees satisfy  $d_{\max}^{\mathrm{in}} \leq \eta_i-1$ and $d_{\max}^{\mathrm{out}} \leq \eta_i-1$, while the largest diagonal entry of $H_i$ satisfies $h_{\max} \leq N-1-\eta_i$. 
    
    Let $\|M_i\|_1 := \max_q \sum_p |[M_i]_{p,q}|$ and $\|M_i\|_\infty := \max_p \sum_q|[M_i]_{p,q}|$ denote the induced $1$- and $\infty$-norms, respectively. Each column of $M_i$ has a diagonal entry bounded by $d_{\max}^{\mathrm{in}} + h_{\max}$, while the sum of the absolute values of its off-diagonal entries is bounded by $d_{\max}^{\mathrm{out}}$. Hence, $\|M_i\|_1 \leq d_{\max}^{\mathrm{in}} + h_{\max} + d_{\max}^{\mathrm{out}} \leq N - 3 + \eta_i$. Similarly, $\|M_i\|_\infty \leq 2d_{\max}^{\mathrm{in}} + h_{\max} \leq N - 3 + \eta_i$. Since $\eta_i \leq N-2$, $\|M_i\|_1 \leq 2N - 5$ and $\|M_i\|_\infty \leq 2N - 5$ hold. Thus, using the standard inequality $\sigma_{\max}(A) \leq \sqrt{\|A\|_1 \, \|A\|_\infty}$, $A \in \mathbb{R}^{n\times n}$, $\sigma_{\max}(M_i) \leq 2N - 5$ is obtained.

    (ii) Note first that, for $N \geq 4$, the right-hand side of \eqref{eq: upper and lower bound on the singular values of M_i}(ii) is smaller than one, since its base $\frac{N-3}{(N-2)(N^2-3N+1)}$ lies in $(0,1)$ and its exponent $\frac{N-3}{2}$ is positive. Hence, if $\eta_i = 1$, the bound holds trivially, because $M_i = H_i$ is a scalar satisfying $M_i \geq 1$ by Lemma~\ref{lemma: neighbors}(ii), and thus $\sigma_{\min}(M_i) \geq 1$. It thus remains to consider $2 \leq \eta_i \leq N-2$.
    According to \cite[Thm.~1]{hong1992lower}, for any $A \in \mathbb{C}^{n \times n}$,
    \begin{equation}
        \label{eq: minimum singular lower bound value for a matrix A}
         \sigma_{\min}(A) \geq \left(\frac{n-1}{n}\right)^{\frac{n-1}{2}} |\det(A)| \max\left\{c(A), r(A) \right\},
    \end{equation}
    where $c(A):= \frac{c_{\min}(A)}{\prod_{l=1}^{n}c_l(A)}$, and $r(A):=\frac{r_{\min}(A)}{\prod_{l=1}^{n}r_l(A)}$, with $c_l(A)$ and $r_l(A)$ denoting the Euclidean norms of the $l$-th column and row of $A$, respectively, and $c_{\min}(A) := \min_l\{c_l(A)\}$ and $r_{\min}(A) := \min_l\{r_l(A)\}$.

    Since $M_i$ has integer entries, $\det(M_i)$ is an integer. Moreover, by Lemma~\ref{Lemma: -M_i is Hirwitz}, all eigenvalues of $M_i$ lie in the open right half-plane. 
    Hence, the real eigenvalues are positive, while each pair of complex conjugate eigenvalues $\lambda_j(M_i),\bar{\lambda}_j(M_i)$ satisfies $\lambda_j(M_i)\bar{\lambda}_j(M_i)=|\lambda_j(M_i)|^2>0$. As $\det(M_i)$ is the product of the eigenvalues, it follows that $\det(M_i)>0$.
    Consequently, $\det(M_i) \geq 1$. 
    Since $\left(\frac{\eta_i-1}{\eta_i}\right)^{(\eta_i-1)/2}$ is decreasing for $\eta_i > 1$, and $\eta_i \leq N-2$ by the definition of $k$-hop out-neighbors, its minimum value is attained for $\eta_i = N-2$. Therefore,
    \begin{equation}
        \label{eq: lower bound on the minimum singular value first bound}
        \sigma_{\min}(M_i) \geq \left(\frac{N-3}{N-2}\right)^{(N-3)/2} \max\left\{c(M_i), r(M_i) \right\}.
    \end{equation}
    To lower-bound $\max\left\{c(M_i), r(M_i) \right\}$ irrespective of the induced topology, we proceed in two steps: we first fix the neighborhood size $\eta_i$ and identify the worst-case topology, and then optimize over $\eta_i$ itself.

    Let $\Outhop{i}{k} = \{q_1,\dots,q_{\eta_i}\} \subseteq
    \mathcal{V}\setminus\{i\}$ denote the global index of the agents in the \mbox{$k$-hop} out-neighborhood of agent $i$. The rows and columns of $M_i$ are ordered according to this labeling, such that the \mbox{$j$-th} row and column correspond to agent $q_j$. 
    
    Similarly to the proof of Lemma~\ref{Lemma: -M_i is Hirwitz}, define $d_{q_j}^{i,\text{in}} := |\inN{q_j} \cap \Outhop{i}{k}|$, and the $j$-th diagonal entry of $H_i$ as $h^i_{q_j}$. Furthermore, let $d_{q_j}^{i,\mathrm{out}} := |\outN{q_j} \cap \Outhop{i}{k}|$. For the \mbox{$j$-th} column and row of $M_i$, we have $c_j(M_i):=\varphi(d_{q_j}^{i,\text{in}}, d_{q_j}^{i,\mathrm{out}},h^i_{q_j}) = \sqrt{(d_{q_j}^{i,\text{in}} + h^i_{q_j})^2 + d_{q_j}^{i,\mathrm{out}}}$ and $r_j(M_i):=\psi(d_{q_j}^{i,\text{in}},h^i_{q_j})= \sqrt{(d_{q_j}^{i,\text{in}} + h^i_{q_j})^2 + d_{q_j}^{i,\text{in}}}$, where $\varphi(\cdot)$ and $\psi(\cdot)$  are \mbox{non-decreasing} functions in each argument. Notice that, by construction, $r_j(M_i)$ is not dependent on $d_{q_j}^{i,\mathrm{out}}$. 

    Since $M_i$ is nonsingular, all its row and column norms are strictly positive.
    Define $\zeta: \mathbb{R}^{\eta_i}_{{>0}}\to\mathbb{R}_{>0}$ as $\zeta(x) :=\frac{\min_l\{x_l\}}{\prod_{j=1}^{\eta_i}x_j}$, so that $c(M_i)=\zeta\big([c_1(M_i),\ldots,c_{\eta_i}(M_i)]^\top\big)$ and $r(M_i)=\zeta\big([r_1(M_i),\ldots,r_{\eta_i}(M_i)]^\top\big)$. Note that
    $\zeta(x) =\min_l\left\{\frac{x_l}{\prod_{j=1}^{\eta_i}x_j} \right\}$. Then, for any $x \in \mathbb{R}^{\eta_i}_{{>0}}$,
    \begin{equation*}
        \zeta(x) =\min_l\left\{\frac{1}{\prod_{j\neq l}x_j}\right\}
        =\frac{1}{\max_l\left\{\prod_{j\neq l}x_j\right\}}.
    \end{equation*}
    Since each product $\prod_{j\neq l}x_j$ is non-decreasing in every component, their maximum is componentwise non-decreasing. Hence, $\zeta$ is componentwise non-increasing, i.e., if $y_l\geq x_l$ for all $l$, then
    \begin{equation}
        \label{eq: monotonicity of zeta}
        \zeta(y)\leq\zeta(x).
    \end{equation}
    We now apply this property to perturbations of the induced topology. Since $\varphi(\cdot)$ is non-decreasing in each argument, increasing a pinning gain $h^i_{q_j}$ does not decrease $c_j(M_i)$. Moreover, adding an edge $q_u\to q_v$, with $q_u,q_v\in\Outhop{i}{k}$, increases $d_{q_u}^{i,\mathrm{out}}$ and $d_{q_v}^{i,\mathrm{in}}$, and therefore does not decrease either $c_u(M_i)$ or $c_v(M_i)$. Consequently, under any sequence of edge additions and pinning-gain increases, the new column norms $c_l^{\mathrm{new}}(M_i)$ satisfy
    \begin{equation}
        c_l^{\mathrm{new}}(M_i)\geq c_l(M_i),
        \qquad l=1,\ldots,\eta_i.
    \end{equation}
    Then, the componentwise monotonicity of $\zeta$ from \eqref{eq: monotonicity of zeta} yields $c^{\mathrm{new}}(M_i)\leq c(M_i)$.
    Likewise, since $\psi(\cdot)$ is non-decreasing in each argument, increasing a pinning gain does not decrease the corresponding $r_j(M_i)$. Adding an edge $q_u\to q_v$ affects only $d_{q_v}^{i,\mathrm{in}}$ among the arguments of $\psi$, and therefore only $r_v(M_i)$ can change, while $r_j(M_i)$ remains unchanged for $j\neq v$. Hence,
    \begin{equation}
        r_l^{\mathrm{new}}(M_i)\geq r_l(M_i),
        \qquad l=1,\ldots,\eta_i,
    \end{equation}
    which implies $r^{\mathrm{new}}(M_i)\leq r(M_i)$. As a result, $\max\{c(M_i),r(M_i)\}$ is non-increasing under any sequence of edge additions and pinning-gain increases. Consequently, for fixed $\eta_i$, the minimum value of $\max\{c(M_i),r(M_i)\}$ is attained when the induced subgraph is complete, for which the corresponding degree and pinning-gain upper bounds are simultaneously attained. In this case, $d_{\max}^{i,\mathrm{in}} = d_{\max}^{i,\mathrm{out}} = \eta_i-1$, and $h^i_{\max}=N-1-\eta_i$.
    
    Since $d_{\max}^{i,\mathrm{in}}+h^i_{\max}=N-2$, all rows and columns of $M_i$ have identical Euclidean norms $c_j(M_i)=r_j(M_i) = \sqrt{(N-2)^2+\eta_i-1} =:c_{\max}(\eta_i)$. Consequently, $c_{\min}(M_i)=r_{\min}(M_i)=c_{\max}(\eta_i)$, and $\prod_{j=1}^{\eta_i}c_j(M_i) = \prod_{j=1}^{\eta_i}r_j(M_i) = c_{\max}(\eta_i)^{\eta_i}.$
    Therefore,
    \begin{equation}
        \label{eq: bound on max c r for fixed eta_i}
        \max\{c(M_i),r(M_i)\} \geq c_{\max}(\eta_i)^{-(\eta_i-1)} .
    \end{equation}
    
    Since $c_{\max}(\eta_i)$ and $\eta_i-1$ are strictly increasing in $\eta_i$, $c_{\max}(\eta_i)^{-(\eta_i-1)}$ is strictly decreasing for $\eta_i > 1$. Hence, its minimum over
    $\eta_i\leq N-2$ is attained at $\eta_i=N-2$, where
    \begin{equation}
        \label{eq: lower bound on ration minimum column, product of column}
        \max\{c(M_i),r(M_i)\} \geq \left(N^2-3N+1\right)^{-(N-3)/2}
    \end{equation}
    Substituting \eqref{eq: lower bound on ration minimum column, product of column} into \eqref{eq: lower bound on the minimum singular value first bound} yields
    \begin{equation*}
        \sigma_{\min}(M_i) \geq \left(\frac{N-3}{(N-2)(N^2-3N+1)}\right)^{(N-3)/2}.
    \end{equation*}
\end{proof}

\subsection{Proof of Proposition~\ref{Proposition: lower bound on the minimum prescribe performance}} \label{Appendix: proof of proposition lower bound on the minimum prescribe performance}
\begin{proof}
    For every $j \in \mathcal{O}_i^{\cup}$, $a_i \le \delta_i^j(0)-\delta_i^j(\infty)$, $b_i \le \delta_i^j(\infty)$, and $\bar l_i \ge l_i^j$.
    Since $a_i \leq \delta_i^j(0)-\delta_i^j(\infty)$ and $e^{-\bar l_i t}\le e^{-l_i^jt}$ for all $t\ge0$, it follows that $a_ie^{-\bar l_it} \leq (\delta_i^j(0)-\delta_i^j(\infty))e^{-l_i^jt}$. Combining this inequality with $b_i\le\delta_i^j(\infty)$ yields $\bar\delta_i(t) \leq \delta_i^j(t)$ for all $t \geq 0$.
    Since the above inequality holds for every $j \in \mathcal{O}_i^{\cup}$, and
    $\Outhop{i}{k}(t) \subseteq \mathcal{O}_i^{\cup}$ for all $t \geq 0$, it follows that $\bar\delta_i(t) \le \min_{j\in\Outhop{i}{k}(t)} \delta_i^j(t)$ for all $t \ge 0$. Finally, because $a_i > 0$, $b_i > 0$, and $\bar l_i > 0$, the function $\bar\delta_i(t)$ is a prescribed performance function.
\end{proof}

\subsection{Proof of Proposition~\ref{Proposition: underline rho lower than rho}}
\label{Appendix: proof of underline rho lower than rho}
\begin{proof}
    Let $s:=t-t_\tau\ge0$. Using the identity $e^{-\bar l_i(t_\tau+s)}=e^{-\bar l_i t_\tau}e^{-\bar l_i s}$, together with $\rho_i^j(t_\tau) = (\rho_i^j(0) -\rho_i^j(\infty))e^{-\bar l_i t_\tau} + \rho_i^j(\infty)$, we obtain $\rho_i^j(t_\tau+s)=(\rho_i^j(t_\tau) -\rho_i^j(\infty))e^{-\bar l_is}+\rho_i^j(\infty)$.
    Furthermore, $ \underline\rho_i^j(t_\tau+s,t_\tau)=(\rho_i^j(t_\tau)-\underline\rho^j_{i,\tau}(\infty))e^{-\underline l^j_{i,\tau}s}+\underline\rho^j_{i,\tau}(\infty)$
    holds by definition of $\underline\rho_i^j(t, t_\tau)$.
    
    For notational convenience, define $c:= \rho_i^j(t_\tau)$, $b_1 := \rho_i^j(\infty)$, $b_2:= \underline\rho^j_{i,\tau}(\infty)$, $l_1 := \bar l_i$ and $l_2 := \underline l^j_{i,\tau}$. Then, 
    \begin{equation*}
        \rho_i^j(t_\tau+s)=(c -b_1)e^{-l_1s}+b_1,
    \end{equation*}
    \begin{equation*}
        \underline\rho_i^j(t_\tau+s,t_\tau)=(c-b_2)e^{-l_2s} + b_2.
    \end{equation*}
    
    Since $l_2 \ge l_1 >0$ and $c \geq b_1 \ge b_2 > 0$, $e^{-l_2 s}\le e^{- l_1 s}$ and $(c-b_2)e^{-l_2 s}\leq (c-b_2)e^{-l_1 s}$ hold for all $s\ge0$. 
    As a result, $ h(s) := \rho_i^j(t_\tau+s)-\underline\rho_i^j(t_\tau+s,t_\tau)$ satisfies
    \begin{equation*}
        h(s) \geq (c - b_1)e^{- l_1 s} - (c-b_2)e^{-l_1 s} + (b_1 - b_2).
    \end{equation*}
    Rearranging the terms yields $h(s) \ge (b_1 - b_2) (1- e^{- l_1 s})$.
    Since $b_1\ge b_2$ and $(1- e^{- l_1 s}) \geq 0$ for all $s\ge0$, it follows that $h(s) \geq 0$. Therefore, $\underline\rho_i^j(t,t_\tau) \leq \rho_i^j(t)$ for all $t \geq t_\tau$.
\end{proof}

\subsection{Proof of Lemma~\ref{Lemma: prescribed performance preservation at switching instants}} \label{Appendix: proof of lemma prescribed performance preservation at switching instants}
\begin{proof}
    Consider an arbitrary agent $i \in \mathcal{V}$. By the triangle and
    submultiplicative inequalities for the induced Euclidean norm, a sufficient
    condition for \eqref{eq: original condition that should hold at the switching instant} is $\norm{M_i(t_{\tau}^+)} \norm{\Xi_i^\tau} \norm{M_i^{-1}(t_{\tau}^-)} \norm{\vect{\xi}_i(t_{\tau}^-)} +  \norm{M_i(t_{\tau}^+)} \norm{\phi_i(t_{\tau})} < \rho_i(t_\tau)$. By the definitions of \(\gamma_i(N)\) and \(\alpha_i(N)\) in
    \eqref{Eq: rho design bound infinite-norm, case random graph} and
    \eqref{eq: alpha(N) definition}, respectively, $\norm{M_i(t_{\tau}^+)} \leq \alpha_i(N)$ and $\norm{M_i^{-1}(t_{\tau}^-)} \leq \gamma_i^{-1}(N)$ hold. Moreover, $\norm{\Xi_i^\tau} \leq 1$ since $\Xi_i^\tau$ is a selection matrix, and $\norm{\phi_i(t_{\tau})} \leq \phi_i^{\max}$ by assumption. Therefore, the sufficient condition becomes
    \begin{equation}
        \label{eq: condition at the swthc after the infinite norm}
        \alpha_i(N) \gamma_i^{-1}(N) \norm{\vect{\xi}_i(t_{\tau}^-)} < \rho_i(t_\tau) - \alpha_i(N) \phi_i^{\max}.
    \end{equation}
    Since, by hypothesis, $|\xi^j_i(t)| < \underline{\rho}_i(t, t_{\tau-1})$ for all $j \in \Outhop{i}{k}$ and $t \in [t_{\tau-1}, t_\tau)$, and the number of components of $\vect{\xi}_i$ satisfies $\eta_i \leq N-2$, it follows that $\norm{\vect{\xi}_i(t_{\tau}^-)} \leq \sqrt{\eta_i} \underline{\rho}_i(t_\tau, t_{\tau-1}) \leq \sqrt{N-2}[({\rho}_{i}(t_{\tau-1}) - \underline{\rho}_{i, \tau-1}(\infty)) e^{- \underline{l}_{i,\tau-1} T_{\tau}} + \underline{\rho}_{i, \tau-1}(\infty)]$. Substituting this bound into \eqref{eq: condition at the swthc after the infinite norm}, with $t_\tau=t_{\tau-1}+T$, and using $\theta_i(N)=\sqrt{N-2}\gamma_i^{-1}(N)\alpha_i(N)$, gives the sufficient condition
    \begin{equation}
        \label{eq: exponential condition explicit T}
        e^{-\underline l_{i,\tau-1}T} <
        \frac{\omega_{i,\tau}(t_{\tau-1}+T,N)}
        {\theta_i(N)\big(\rho_i(t_{\tau-1})-\underline\rho_{i,\tau-1}(\infty)\big)}.
    \end{equation}
    By assumption, $\omega_{i,\tau}(\infty,N)>0$, while $\rho_i(t_{\tau-1})>\underline\rho_{i,\tau-1}(\infty)$ is ensured by design. Thus, since $\omega_{i,\tau}(t,N)\geq \omega_{i,\tau}(\infty,N)$ for all $t$, both sides of \eqref{eq: exponential condition explicit T} are strictly positive, as established in Remark~\ref{remark: well-poseness of F}. Taking the logarithms yields the equivalent condition $T > F_{i,\tau}(T)$, where $F_{i,\tau}$ is defined as per \eqref{eq: def of F(T)}.    

    It remains to show that the dwell-time condition guarantees the above
    inequality.
    Since $\omega_{i,\tau}(t,N)$ is non-increasing in $t$, $F_{i,\tau}(T)$ is non-decreasing in $T$ and satisfies
    \begin{equation}
        \label{eq: F bounded by closed form}
        F_{i,\tau}(T) \leq F_{i,\tau}^\infty
    \end{equation}
    with $F_{i,\tau}^\infty  := \frac{1}{\underline l_{i,\tau-1}}\ln\left(\frac{\theta_i(N)\big(\rho_i(t_{\tau-1})-\underline\rho_{i,\tau-1}(\infty)\big)}{\omega_{i,\tau}(\infty,N)}\right)$.
    
    Define the set $\mathcal S_i := \left\{T\geq0:T'>F_{i,\tau}(T') \;\forall T'\geq T \right\}$, so that $T_{i,\tau}^{\min}=\inf\mathcal S_i$.
    For every $T_0>F_{i,\tau}^{\infty}$ and every $T'\geq T_0$, it holds that $T' \geq T_0 > F_{i,\tau}^{\infty} \geq F_{i,\tau}(T')$, and hence $T_0\in\mathcal S_i$. Consequently, $\mathcal S_i$ is non-empty and $T_{i,\tau}^{\min} \leq F_{i,\tau}^{\infty} < \infty$.
    
    Assume now that $T_\tau > \max_{i \in \mathcal{V}} T_{i,\tau}^{\min}$. Since $T_{i,\tau}^{\min} =\inf \mathcal S_i$ and $T_\tau>T_{i,\tau}^{\min}$, the definition of the infimum guarantees the existence of some $T^\ast\in\mathcal S_i$ such that $T^\ast<T_\tau$.
    Because $T^\ast\in\mathcal S_i$, the inequality $T'>F_{i,\tau}(T')$ holds for every $T'\geq T^\ast$. In particular, choosing $T'=T_\tau$ gives $T_\tau>F_{i,\tau}(T_\tau)$.
    
    Because $T_\tau>F_{i,\tau}(T_\tau)$, $\underline{l}_{i, \tau-1}> 0$, and the argument of the logarithm in \eqref{eq: def of F(T)} is positive, the definition of $F_{i,\tau}$ implies
    \begin{equation*}
        e^{-\underline l_{i,\tau-1}T_\tau}
        <
        \frac{\omega_{i,\tau}(t_\tau,N)}
        {\theta_i(N)\big(\rho_i(t_{\tau-1})-\underline\rho_{i,\tau-1}(\infty)\big)},
    \end{equation*}
    which is \eqref{eq: exponential condition explicit T} evaluated at $T=T_\tau$, with $t_\tau=t_{\tau-1}+T_\tau$. Multiplying both sides by the positive denominator, substituting the definition of $\omega_{i,\tau}(t_\tau, N)$, and rearranging gives
    \begin{equation*}
        \theta_i(N)\underline\rho_i(t_\tau,t_{\tau-1}) < \rho_i(t_\tau)-\alpha_i(N)\phi_i^{\max},
    \end{equation*}
    where $\underline\rho_i(t_\tau,t_{\tau-1}) = \big[\big(\rho_i(t_{\tau-1})-\underline\rho_{i,\tau-1}(\infty)\big)e^{-\underline l_{i,\tau-1}T_\tau}+\underline\rho_{i,\tau-1}(\infty)\big]$. Since \eqref{eq: final constraint to be satisfied by the disagreement vector} holds by hypothesis for all $t \in [t_{\tau-1}, t_\tau)$ with $\underline{\rho}^{j}_i(t, t_{\tau})$ replaced by $ \underline\rho_i(t,t_{\tau-1})$, it follows that $\norm{\vect{\xi}_i(t^-_\tau)} \leq \sqrt{N-2} \underline\rho_i(t_\tau,t_{\tau-1})$. Consequently, since $\theta_i(N)=\sqrt{N-2}\,\gamma_i^{-1}(N)\alpha_i(N)$, \eqref{eq: condition at the swthc after the infinite norm} is satisfied at the switching instant $t_\tau$.

    Since $i$ was arbitrary and the \mbox{dwell-time} condition is imposed for all
    agents simultaneously, \eqref{eq: general condition to be satisfied at the stime of switch} holds for every $i \in \mathcal{V}$.
    
\end{proof}

\bibliographystyle{IEEEtran}
\bibliography{reference}

\end{document}